\documentclass[lettersize,journal]{IEEEtran}

\usepackage{amsthm}

\usepackage{amsmath,amsfonts,amssymb}
\usepackage{bm}
\usepackage{mathtools}

\usepackage{graphicx}
\usepackage{booktabs}
\usepackage{multirow}
\usepackage[caption=false,font=normalsize,labelfont=sf,textfont=sf]{subfig}
\usepackage{stfloats}
\usepackage{array}

\usepackage{cite}
\usepackage{url}
\usepackage{textcomp}
\usepackage{color}
\usepackage{tikz}
\usepackage{enumitem}
\usetikzlibrary{arrows.meta,positioning,decorations.pathreplacing,calc}

\newtheorem{theorem}{Theorem}
\newtheorem{lemma}[theorem]{Lemma}
\newtheorem{proposition}[theorem]{Proposition}
\newtheorem{corollary}[theorem]{Corollary}
\newtheorem{remark}{Remark}
\newtheorem{definition}{Definition}
\newtheorem{assumption}{Assumption}

\newcommand{\E}{\mathbb{E}}

\newcommand{\SNR}{\mathrm{SNR}}
\newcommand{\CRB}{\mathrm{CRB}}

\newcommand{\lequiv}{l^{\mathrm{eq}}}

\begin{document}

\title{A Cancellation Mechanism in AFDM Radar Sensing:\\ Exact Fisher Information and Delay-Doppler Decoupling}

\author{Tingjun Lyu,
        Yunmei Shi
\thanks{Manuscript received XXXX; revised XXXX. This work was supported in part by XXXX.}%
\thanks{Tingjun Lyu and Yunmei Shi are with the Department of Communication and Information Engineering, Tongji University, Shanghai, China (e-mail: 2531934@tongji.edu.cn, ymshi@tongji.edu.cn).}%
}

\markboth{IEEE Transactions on Signal Processing, Vol.~XX, No.~XX, XXXX~202X}%
{Author \MakeLowercase{\textit{et al.}}: A Cancellation Mechanism in AFDM Radar Sensing}

\maketitle

\begin{abstract}
We consider radar sensing with affine frequency division multiplexing (AFDM), a chirp-based waveform recently proposed for high-mobility integrated sensing and communication. While numerical Cram\'{e}r-Rao bounds for AFDM radar are available in the literature, no closed-form Fisher information analysis has so far revealed how the waveform's chirp structure shapes delay-Doppler estimation accuracy.In this paper, we provide such an analysis. We identify a cancellation in the AFDM likelihood: the frequency drift introduced by the chirp modulation is exactly compensated by a discrete phase correction built into the chirp-periodic prefix, leaving only a small residual. Exploiting this cancellation, we derive an exact closed-form Fisher information matrix that depends on the AFDM chirp structure through a single scalar, and from it we obtain closed-form Cram\'{e}r-Rao bounds for joint delay and Doppler estimation.Three consequences follow. AFDM is provably less delay-Doppler-coupled than OFDM for any nonzero chirp rate. The delay Cram\'{e}r-Rao bound improves quadratically with the chirp rate, while the Doppler bound is unaffected by it. Finally, our framework reduces continuously to the classical OFDM result as the chirp vanishes, certifying it as a strict generalization of OFDM radar sensing theory.Overall, our work shows that the chirp-periodic prefix---until now studied only as a channel-equalization device---is the structural element that decouples delay and Doppler in AFDM sensing, and that AFDM's superior sensing performance can be characterized analytically rather than through numerical bounds alone. Numerical experiments at realistic vehicular and low-Earth-orbit parameters validate all closed-form expressions.
\end{abstract}


\section{Introduction}
\label{sec:intro}

\IEEEPARstart{I}{ntegrated} sensing and communication (ISAC) has emerged as a pivotal research direction for sixth-generation (6G) wireless networks~\cite{liu2022survey,liu2022limits}, but high-mobility deployments such as vehicular-to-everything (V2X) and low-altitude unmanned aerial vehicle (UAV) operations push conventional orthogonal frequency division multiplexing (OFDM) waveforms past their useful range: the significant Doppler shifts induced by target motion cause severe inter-carrier interference (ICI) that simultaneously degrades communication reliability and sensing accuracy~\cite{sturm2011ofdm}. Designing a waveform whose sensing limits remain favorable in this regime is therefore a central problem for 6G ISAC.

The theoretical foundations of multicarrier radar sensing have been laid out for OFDM and OTFS. Gaudio~\emph{et al.}~\cite{gaudio2019ofdm} derived the Fisher information matrix and CRB for OFDM-based joint delay-Doppler estimation and showed that OFDM degrades significantly at high Doppler. Keskin~\emph{et al.}~\cite{keskin2021ofdm} subsequently developed CRB-optimal waveform design for OFDM dual-functional radar-communications. The orthogonal time-frequency-space (OTFS) waveform, operating in the delay-Doppler domain, has been shown to achieve near-optimal radar estimation under high mobility~\cite{gaudio2019ofdm,raviteja2019radar}. Yet OFDM and OTFS face well-known limitations in the high-mobility regime: OFDM suffers from inter-carrier interference, while OTFS relies on a two-dimensional symplectic Fourier transform whose sensing-theoretic properties are tightly bound to its specific guard structure~\cite{rou2024spm}.

Affine frequency division multiplexing (AFDM)~\cite{bemani2021afdm,bemani2023twc} is a recently proposed multicarrier waveform that addresses the high-mobility challenge through a fundamentally different approach. AFDM modulates information onto a discrete affine Fourier transform (DAFT) basis whose subcarriers carry a tunable quadratic phase, and protects the signal with a chirp-periodic prefix (CPP) that plays the role of the cyclic prefix in OFDM. With chirp parameters matched to the channel, the CPP enables full-diversity reception over doubly dispersive channels~\cite{bemani2023twc}. From a sensing standpoint, AFDM occupies a unique position: unlike classical chirp radar (e.g., LFM and FMCW), AFDM is a multicarrier waveform that supports communication and sensing within a single frame; unlike OFDM, its delay-Doppler response is shaped by chirp modulation rather than by pure subcarrier orthogonality. This combination makes AFDM a strong candidate for high-mobility ISAC and motivates a careful analysis of its fundamental sensing limits.

\subsection{Related Work on AFDM Sensing Performance}

The sensing potential of AFDM has begun to attract attention, and existing work falls into two broad categories. \emph{Algorithmic studies} focus on estimation algorithms and pilot design. Ni~\emph{et al.}~\cite{ni2022afdm} first studied AFDM-based radar parameter estimation, designing both time-domain and DAFT-domain algorithms; however, the analysis is empirical and does not characterize the achievable estimation accuracy in closed form. Bemani and Kountouris~\cite{bemani2024isac} demonstrated that a single DAFT-domain pilot can achieve near-full-frame sensing performance and that AFDM's chirp structure enables simplified self-interference cancellation for monostatic ISAC; the focus, however, is on receiver architecture rather than on fundamental performance limits. \emph{Performance-bound studies} target the CRB. Ranasinghe~\emph{et al.}~\cite{ranasinghe2025joint} developed joint channel, data, and radar parameter estimation algorithms across OFDM, OTFS, and AFDM, providing numerical CRB comparisons but no closed-form expressions. Most recently, Zhang~\emph{et al.}~\cite{zhang2025afdm} derived CRBs for AFDM target detection and analyzed the AFDM ambiguity function, but their results do not expose the internal structure of the delay score function nor the mechanism by which the CPP shapes delay-Doppler estimability. Across these works, AFDM's chirp structure is treated either through numerical experiments or through the OFDM special case as a black box, leaving the fundamental sensing-theoretic role of the CPP unaddressed.

\subsection{Motivation and Contributions}

Despite this growing body of work, no closed-form Fisher information analysis of AFDM radar sensing has so far revealed how the chirp structure shapes delay-Doppler estimation, why AFDM's delay-Doppler coupling is structurally weaker than OFDM's, or how this advantage scales with the chirp rate. The present paper provides such an analysis. The key observation, exploited throughout, is that the AFDM likelihood contains a structural cancellation: the frequency drift introduced by the chirp modulation is exactly compensated by the phase jump of the CPP. This cancellation, made precise in Section~\ref{sec:derivatives}, is what allows the entire Fisher information matrix to be expressed in closed form through a single scalar. The main contributions are summarized below.

\begin{itemize}
\item We identify a \emph{cancellation mechanism} in the AFDM delay score: the four paths through which the delay enters the AFDM likelihood reduce, after exact cancellation between the chirp-induced frequency drift and the CPP phase jump, to a single residual term of order $1/C$. To the best of our knowledge, this is the first time such a structural cancellation has been identified in a chirp-based multicarrier waveform.

\item Exploiting the cancellation, we derive an exact closed-form $2\times 2$ Fisher information matrix for joint delay-Doppler estimation. The entire dependence on the AFDM chirp structure collapses into a single scalar $\eta$, reducing what would otherwise be a multi-parameter problem to a one-dimensional one. Setting $\eta=1$ recovers the OFDM Fisher information.

\item We prove that the delay-Doppler coupling coefficient of AFDM is strictly smaller than that of OFDM for any nonzero chirp rate. The result is structural, not parameter-tuning: AFDM's coupling advantage holds unconditionally and is fully captured by $\eta$.

\item We obtain closed-form Cram\'{e}r-Rao bounds in two complementary frameworks (channel-known and phase-profiled), with Schur complementation yielding an exact delay-Doppler decoupling in the latter. The delay bound improves quadratically with the chirp rate while the Doppler bound is unaffected by it, giving waveform designers a clean dimensional separation. Our bounds reduce continuously to the classical OFDM CRB of~\cite{gaudio2019ofdm} as the chirp vanishes.

\item We identify the operating regime in which AFDM's analytical sensing gain becomes practically useful. Numerical experiments at realistic V2X and low-Earth-orbit parameters show that the predicted improvement over OFDM materializes already at the smallest communication-compatible chirp rate---approximately a four-fold delay-CRB gain at the standard configuration $C = 2\alpha_{\max}{+}1$---and that all closed-form expressions match numerical Fisher information evaluation across the full chirp-rate range $c_1 \in (0, 1/2)$.
\end{itemize}

The remainder of this paper is organized as follows. Section~\ref{sec:system_model} introduces the AFDM signal model and the DAFT-domain sensing observation. Section~\ref{sec:derivatives} establishes the cancellation mechanism by decomposing the delay score into four paths and identifying the two that cancel. Section~\ref{sec:fim} exploits this cancellation to derive the closed-form Fisher information matrix captured by the single scalar $\eta$. Section~\ref{sec:coupling} proves that the resulting delay-Doppler coupling coefficient $3/(4\sqrt{\eta})$ is strictly smaller than the OFDM value. Section~\ref{sec:crb} derives the exact CRBs, obtains the $c_1^{-2}$ dimensional gain, and establishes the continuous reduction to the classical OFDM bound via Schur complementation. Section~\ref{sec:simulation} provides numerical validation. Section~\ref{sec:conclusion} concludes the paper. Detailed proofs are collected in Appendices~\ref{app:derivatives}--\ref{app:eta}.

\emph{Notation:}
Boldface lowercase and uppercase letters denote vectors and matrices, respectively. $(\cdot)^T$, $(\cdot)^H$, and $(\cdot)^{-1}$ denote transpose, conjugate transpose, and matrix inverse. $\E[\cdot]$ denotes expectation, and $\Re\{\cdot\}$ and $\Im\{\cdot\}$ denote real and imaginary parts. $j=\sqrt{-1}$ is the imaginary unit. $\delta[\cdot]$ denotes the Kronecker delta. $\lfloor\cdot\rfloor$ and $\lceil\cdot\rceil$ denote floor and ceiling functions. $\mathcal{CN}(0,\sigma^2)$ denotes a circularly symmetric complex Gaussian distribution with zero mean and variance $\sigma^2$. $(\cdot)_N$ denotes the modulo-$N$ operation.

\section{System Model}
\label{sec:system_model}

We consider a AFDM based ISAC system. Let $\mathbf{x}=[x[0],x[1],\ldots,x[N-1]]^T$ denote the $N$-point discrete affine Fourier transform (DAFT)-domain transmit symbol vector. The time-domain signal can be readily  obtained via the inverse DAFT (IDAFT):
\begin{equation}
s[n] = \frac{1}{\sqrt{N}} \sum_{m=0}^{N-1} x[m] \cdot \Psi_n(m), \quad n = 0,1,\ldots,N-1
\label{eq:idaft}
\end{equation}
where the DAFT basis function is
\begin{equation}
\Psi_n(m) = e^{j2\pi\left(c_2 m^2 + \frac{mn}{N} + c_1 n^2 + \Phi_{\mathrm{seg}}(n)\right)}.
\label{eq:daft_basis}
\end{equation}
Here $c_1$ and $c_2$ are chirp parameters, and $\Phi_{\mathrm{seg}}(n)$ is an auxiliary phase term that some AFDM variants use to refine the chirp envelope across symbol boundaries. The standard AFDM configuration analyzed in this paper, following~\cite{bemani2023twc}, uses $c_2 = -c_1$ and $\Phi_{\mathrm{seg}}(n) \equiv 0$, so that the DAFT basis simplifies to
\begin{equation}
\Psi_n(m) = e^{j2\pi\left(-c_1 m^2 + \frac{mn}{N} + c_1 n^2\right)}.
\label{eq:daft_basis_standard}
\end{equation}
When $c_1 = 0$, \eqref{eq:daft_basis_standard} reduces to the standard DFT kernel and AFDM degenerates to OFDM.

\subsection{Sensing Channel Model}

We model the radar sensing channel as
\begin{equation}
h(t, \tau) = h \cdot e^{j2\pi f_D t} \cdot \delta(\tau - \tau_0)
\label{eq:channel}
\end{equation}
where $h = \alpha e^{j\phi}$ is the complex channel gain, $\tau_0$ is the round-trip delay, and $f_D$ is the Doppler frequency shift. This way, the radar sensing parameter set can be represented as 
\begin{equation}
	\bm{\theta} = (\alpha, \phi, \tau_0, f_D).
	\label{eq:parameter_vector}
\end{equation}
Note that the continuous parameters are mapped to the sampling domain via
\begin{align}
\tau_0 &= (l + \iota)\Delta t, \quad l \in \mathbb{Z}_{\geq 0}, \; \iota \in (-\tfrac{1}{2}, \tfrac{1}{2}] \label{eq:delay_norm} \\
T f_D &= k + \kappa, \quad k \in \mathbb{Z}, \; \kappa \in (-\tfrac{1}{2}, \tfrac{1}{2}] \label{eq:doppler_norm}
\end{align}
where $l$ and $k$ are the integer parts, $\iota$ and $\kappa$ are the fractional parts, $\Delta t$ denotes the sampling duration, and $T = N\Delta t$ is the frame duration. The integer delay satisfies $l \leq \alpha_{\max}$, where $\alpha_{\max}$ is the maximum delay tap supported by the CPP guard interval~\cite{bemani2023twc}.

\subsection{DAFT-Domain Input-Output Relation}

After channel propagation, sampling, cyclic prefix removal, and forward DAFT processing, the input-output relation in the DAFT domain can be expressed as \cite{bemani2023twc}
\begin{equation}
y[m] = \frac{h}{N} \sum_{m'=0}^{N-1} x[m'] \, e^{j2\pi \Xi(m, m')} \cdot \mathcal{F}(m, m') + w[m]
\label{eq:io_relation}
\end{equation}
where $w[m]\sim\mathcal{CN}(0,N_0)$ is additive white Gaussian noise.

\subsubsection{Phase Exponent $\Xi(m,m')$}
The phase exponent is given by
\begin{equation}
\Xi(m, m') = c_1(l + \iota)^2 - c_2(m^2 - m'^2) - \frac{(l + \iota) m'}{N}.
\label{eq:Xi_def}
\end{equation}
A crucial property is that $\Xi$ depends \emph{only} on the delay parameters $(l,\iota)$ and does not involve the Doppler parameter $f_D$. This asymmetry plays a central role in the subsequent FIM analysis.

\subsubsection{Equivalent delay, response kernel, and CPP structure}
The response function is
\begin{equation}
\mathcal{F}(m, m') = \sum_{n=0}^{N-1} \Gamma_n(m, m'; \tau_0, f_D)
\label{eq:F_def}
\end{equation}
where the segmented chirp response term is
\begin{equation}
\Gamma_n = e^{j\frac{2\pi(m'-(m+\lequiv))n}{N}} \cdot e^{j\Phi_{\mathrm{jump}}(n,m)}.
\label{eq:Gamma_n}
\end{equation}
The first factor is a Dirichlet-type kernel centered at $m + \lequiv$, where the \emph{equivalent delay}
\begin{equation}
\lequiv \triangleq (k + \kappa) + 2Nc_1(l + \iota)
\label{eq:leq_def}
\end{equation}
couples the delay and Doppler parameters---it is the central quantity shaping the DAFT-domain peak location. The second factor of $\Gamma_n$ contains the phase jump
\begin{equation}
\Phi_{\mathrm{jump}}(n, m) = 2\pi\iota \cdot Q(n,m),
\label{eq:phi_jump_def}
\end{equation}
in which $Q(n,m)$ is the staircase function
\begin{equation}
Q(n,m) = \sum_{q=0}^{C} q \cdot \mathcal{I}_{L_{m,q}}\!\left((n - (l+\iota))_N\right),
\label{eq:Q_def}
\end{equation}
where $C = 2c_1 N$ is the number of chirp wrapping cycles, $L_{m,q}$ denotes the $q$-th time-domain segment, $\mathcal{I}_{L_{m,q}}(\cdot)$ is its indicator function, and $(\cdot)_N$ is the modulo-$N$ operation. The staircase $Q(n,m)$ partitions the time axis $[0,N)$ into $C+1$ segments of approximate width $W = 1/(2c_1)$, taking the value $q$ on the $q$-th segment. The CPP ensures perfect channel equalization in the DAFT domain when $C$ is a positive integer~\cite{bemani2023twc}; as shown in Section~\ref{sec:derivatives}, $\Phi_{\mathrm{jump}}$ also plays the key role of cancelling the chirp-induced frequency drift in the delay score, which is the central structural mechanism of this paper.

\section{A Cancellation Mechanism in the Delay Score}
\label{sec:derivatives}

This section establishes the central structural observation of the paper. We show that the delay score---the partial derivative of the AFDM log-likelihood with respect to the delay parameter---naturally decomposes into four physically distinct paths, and that two of these paths cancel up to an $O(1/C)$ residual. This cancellation is the algebraic manifestation of the CPP compensating the chirp-induced frequency drift, and it is the mechanism that makes all the closed-form results in Sections~\ref{sec:fim}--\ref{sec:crb} possible. Fig.~\ref{fig:mechanism_overview} summarizes the mechanism and its downstream consequences.

\begin{figure*}[t]
\centering
\begin{tikzpicture}[
  font=\small,
  >={Stealth[length=2.2mm]},
  pathbox/.style={draw, rounded corners=2pt, align=center, inner sep=3pt, minimum width=28mm, minimum height=8mm},
  mechbox/.style={draw, thick, rounded corners=3pt, align=center, inner sep=4pt, minimum width=36mm, minimum height=14mm, fill=gray!10},
  conbox/.style={draw, rounded corners=3pt, align=center, inner sep=3pt, font=\small, minimum width=42mm, minimum height=10mm, fill=green!5},
  driftarrow/.style={->, thick, red!75!black, line width=0.7pt},
  comparrow/.style={->, thick, blue!60!black, line width=0.7pt},
  bypassarrow/.style={->, thick, gray!55, dashed},
  conarrow/.style={->, thick, black!55}
]
\node[pathbox]               (p1)  at (0, 2.7) {Path~1: $\Xi$ phase};
\node[pathbox, fill=red!10]  (p2)  at (0, 1.2) {Path~2: $-2c_1 n$};
\node[pathbox, fill=blue!10] (p3a) at (0,-0.3) {Path~3a: $+Q(n,m)$};
\node[pathbox]               (p3b) at (0,-1.8) {Path~3b: boundary};

\node[mechbox] (cancel) at (6.3, 0.45)
   {\textbf{Cancellation (Theorem~\ref{thm:cancellation})}\\[2pt]
    $-2c_1 n + Q(n,m) = -2c_1(l{+}\iota) + \epsilon_Q$};

\draw[driftarrow] (p2.east)  -- node[above, font=\scriptsize, red!75!black]  {$-$\,chirp drift} (cancel.west |- p2.east);
\draw[comparrow]  (p3a.east) -- node[below, font=\scriptsize, blue!60!black] {$+$\,CPP compensation} (cancel.west |- p3a.east);

\draw[bypassarrow] (p1.east) -- ++(0.5,0) |- (cancel.north west);
\draw[bypassarrow] (p3b.east) -- ++(0.5,0) |- (cancel.south west);

\node[conbox] (c1) at (12.5, 2.45) {Closed-form FIM\\ parameterized by a single $\eta$};
\node[conbox] (c2) at (12.5, 0.45) {$\rho = 3/(4\sqrt{\eta}) < 3/4$\\ weaker than OFDM};
\node[conbox] (c3) at (12.5,-1.55) {Delay CRB $\propto c_1^{-2}$;\\ Doppler CRB $c_1$-free};

\draw[conarrow] (cancel.east) -- (c1.west);
\draw[conarrow] (cancel.east) -- (c2.west);
\draw[conarrow] (cancel.east) -- (c3.west);

\node[font=\footnotesize\bfseries, gray!55!black] at (0,   3.7) {Four-Path Decomposition};
\node[font=\footnotesize\bfseries, gray!55!black] at (6.3, 3.7) {Mechanism};
\node[font=\footnotesize\bfseries, gray!55!black] at (12.5,3.7) {Closed-Form Consequences};

\node[font=\scriptsize, gray!70!black] at (0,   -2.8) {Section~\ref{sec:derivatives}};
\node[font=\scriptsize, gray!70!black] at (6.3, -2.8) {Theorem~\ref{thm:cancellation}};
\node[font=\scriptsize, gray!70!black] at (12.5,-2.8) {Sections~\ref{sec:fim}--\ref{sec:crb}};

\begin{scope}[shift={(-0.4,-4.0)}]
  \node[draw, rounded corners=2pt, inner sep=3pt, font=\scriptsize, align=left, minimum width=62mm] (leg) at (3.2,0) {
    \textcolor{red!75!black}{\textbf{---}} chirp-induced drift \quad
    \textcolor{blue!60!black}{\textbf{---}} CPP compensation \quad
    \textcolor{gray!55}{\textbf{- -}} bypass (no cancellation role)
  };
\end{scope}
\end{tikzpicture}
\caption{The cancellation mechanism at a glance. The delay score decomposes into four paths; Path~2 (chirp-induced frequency drift, red) and Path~3a (CPP phase-jump compensation, blue) have equal magnitude and opposite sign, and cancel up to an $O(1/C)$ residual. Paths~1 and~3b bypass the cancellation and contribute separately. Every closed-form result in Sections~\ref{sec:fim}--\ref{sec:crb} is a consequence of this cancellation.}
\label{fig:mechanism_overview}
\end{figure*}
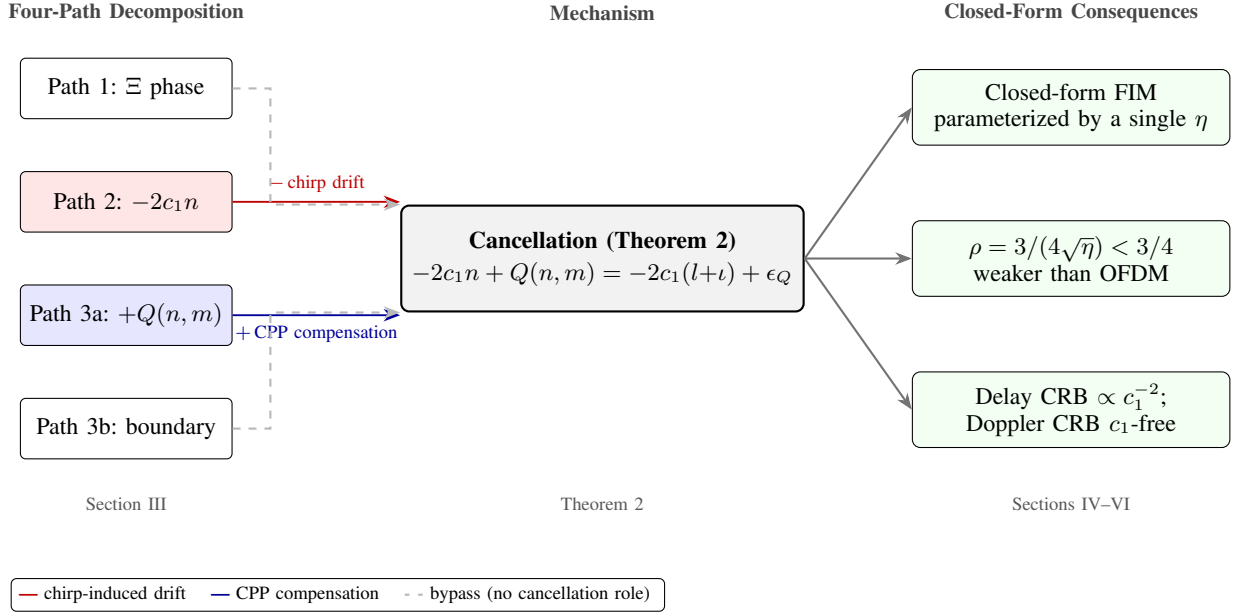

The Cram\'{e}r-Rao bound analysis requires the partial derivatives $\partial y[m]/\partial\theta_i$ for each parameter $\theta_i \in \{\alpha, \phi, \tau_0, f_D\}$. While the derivatives with respect to $\alpha$, $\phi$, and $f_D$ admit a single-path structure, the delay derivative $\partial y[m]/\partial\tau_0$ exhibits a nontrivial four-path dependency whose analysis reveals the cancellation.

\subsection{Derivative Strategy}

From the observation model \eqref{eq:io_relation}, each parameter influences $y[m]$ through different intermediate quantities:
\begin{itemize}
\item $\alpha, \phi$: only through $h = \alpha e^{j\phi}$ (single path);
\item $\tau_0$: through $\Xi(m,m')$, $\lequiv$, and $\Phi_{\mathrm{jump}}$ (multiple paths, requiring the chain and product rules);
\item $f_D$: only through $\lequiv$ via the term $Tf_D$ (single path).
\end{itemize}
For the multi-path dependency of $\tau_0$, the multivariate chain rule is applied systematically. The key principle is that the summation indices $m'$ and $n$ must be retained throughout, and the product $e^{j2\pi\Xi}\cdot\mathcal{F}$ requires the product rule since both factors may depend on $\tau_0$. A critical observation is the \emph{product structure} of $\Phi_{\mathrm{jump}} = 2\pi\iota \cdot Q$: both the scaling factor $\iota$ and the staircase $Q$ depend on $\tau_0$, so differentiating $\Phi_{\mathrm{jump}}$ requires the product rule and generates two distinct contributions (Paths~3a and~3b below).

\subsection{Derivatives with Respect to $\alpha$ and $\phi$}

Since $\alpha$ and $\phi$ enter only through $h = \alpha e^{j\phi}$, the derivatives are straightforward:
\begin{align}
\frac{\partial y[m]}{\partial\alpha} &= \frac{e^{j\phi}}{N}\sum_{m'=0}^{N-1} x[m'] \, e^{j2\pi\Xi(m,m')}\mathcal{F}(m,m'), \label{eq:dy_dalpha} \\
\frac{\partial y[m]}{\partial\phi} &= \frac{j\alpha e^{j\phi}}{N}\sum_{m'=0}^{N-1} x[m'] \, e^{j2\pi\Xi(m,m')}\mathcal{F}(m,m'). \label{eq:dy_dphi}
\end{align}
The two derivatives satisfy $\partial y/\partial\phi = j\alpha\cdot\partial y/\partial\alpha$, which implies a specific diagonal structure in the $(\alpha, \phi)$ sub-block of the FIM. The analysis hereafter focuses on the $(\tau_0, f_D)$ sub-block.

\subsection{Delay Derivative: Four-Path Decomposition}
\label{subsec:deriv_tau}

The delay $\tau_0$ influences $y[m]$ through the normalized delay $l + \iota = \tau_0/\Delta t$ via four distinct paths:
\begin{equation}
\tau_0 \to (l{+}\iota) \to
\begin{cases}
\text{Path~1}: \;\;\, \Xi(m,m') \to e^{j2\pi\Xi} \to y[m] \\
\text{Path~2}: \;\;\, \lequiv \to \Gamma_n \to \mathcal{F} \to y[m] \\
\text{Path~3a}: \, 2\pi\iota \cdot Q \to \Gamma_n \to \mathcal{F} \to y[m] \\
\text{Path~3b}: \, n_q \to Q \to \Gamma_n \to \mathcal{F} \to y[m]
\end{cases}
\label{eq:four_paths}
\end{equation}
where Path~1 acts through the global $\Xi$-phase factor in~\eqref{eq:io_relation}, Path~2 through the Dirichlet kernel shift inside $\Gamma_n$, Path~3a through the $\iota$-scaling of $\Phi_{\mathrm{jump}}$, and Path~3b through the boundary migration of the staircase. Paths~3a and~3b both arise from applying the product rule to $\Phi_{\mathrm{jump}} = 2\pi\iota \cdot Q$, since both $\iota$ and the segment boundaries of $Q$ depend on $\tau_0$.

\subsubsection{Path~1: Derivative of $\Xi(m,m')$}

Differentiating $\Xi$ in \eqref{eq:Xi_def} with respect to $\tau_0$ via $l+\iota = \tau_0/\Delta t$ yields
\begin{equation}
\frac{\partial\Xi}{\partial\tau_0} = \frac{2c_1(l+\iota)}{\Delta t} - \frac{m'}{N\Delta t}.
\label{eq:dXi_dtau}
\end{equation}
The first term originates from the quadratic phase $c_1(l{+}\iota)^2$ and is independent of $m'$; the second term from the linear phase $-(l{+}\iota)m'/N$ varies linearly with $m'$. The $m'$-linear component becomes the dominant contribution to $I_{\tau\tau}$ after cancellation.

\subsubsection{Paths~2 and~3: Derivative of $\Gamma_n$}

Since $\Gamma_n = e^{j\varphi_n}$ with total phase
\begin{equation}
\varphi_n = \underbrace{\frac{2\pi(m'{-}m{-}\lequiv)}{N}\cdot n}_{\text{Dirichlet kernel phase}} + \underbrace{2\pi\iota\cdot Q(n,m)}_{\Phi_{\mathrm{jump}}},
\label{eq:varphi_n}
\end{equation}
the derivative is $\partial\Gamma_n/\partial\tau_0 = j(\partial\varphi_n/\partial\tau_0)\cdot\Gamma_n$. Computing $\partial\varphi_n/\partial\tau_0$ yields three contributions:

\paragraph{Path~2 (Dirichlet kernel shift)} Using $\partial\lequiv/\partial\tau_0 = 2Nc_1/\Delta t$ from \eqref{eq:leq_def}:
\begin{equation}
\frac{\partial\varphi_n}{\partial\tau_0}\bigg|_{\text{Path~2}} = -\frac{4\pi c_1 n}{\Delta t}.
\label{eq:path2}
\end{equation}

\paragraph{Path~3a ($\iota$-scaling)} Applying the product rule to $\Phi_{\mathrm{jump}} = 2\pi\iota\cdot Q$ and differentiating $\iota$ (with $\partial\iota/\partial\tau_0 = 1/\Delta t$ in a neighborhood where $\tau_0$ does not cross an integer boundary):
\begin{equation}
\frac{\partial\varphi_n}{\partial\tau_0}\bigg|_{\text{Path~3a}} = \frac{2\pi}{\Delta t}Q(n,m).
\label{eq:path3a}
\end{equation}

\paragraph{Path~3b (boundary migration)} The second term of the product rule differentiates $Q$ with respect to its boundary locations:
\begin{equation}
\frac{\partial\varphi_n}{\partial\tau_0}\bigg|_{\text{Path~3b}} = \frac{2\pi\iota}{\Delta t}\sum_{q=0}^{C} q\cdot\frac{\partial\mathcal{I}_{L_{m,q}}}{\partial(l{+}\iota)}.
\label{eq:path3b}
\end{equation}
This term acts only at the $C$ segment boundaries and has magnitude $O(\iota)$, much smaller than Paths~2 and~3a.

Combining the three contributions:
\begin{equation}
\frac{\partial\Gamma_n}{\partial\tau_0} = \Gamma_n\cdot\frac{j2\pi}{\Delta t}\left[\underbrace{-2c_1 n}_{\text{Path~2}} + \underbrace{Q(n,m)}_{\text{Path~3a}} + \underbrace{\iota\frac{\partial Q}{\partial(l{+}\iota)}}_{\text{Path~3b}}\right].
\label{eq:dGamma_dtau}
\end{equation}

\subsubsection{Complete Delay Derivative}

Applying the product rule to $e^{j2\pi\Xi}\cdot\mathcal{F}$ in \eqref{eq:io_relation} and combining Path~1 with Paths~2--3:
\begin{equation}
\begin{split}
\frac{\partial y[m]}{\partial\tau_0}
&= \frac{h}{N}\sum_{m'} x[m'] \, e^{j2\pi\Xi} \\
&\qquad\times\bigl[A(m')\mathcal{D}_{\Xi}
      + B_{\mathrm{s}}\mathcal{D}_{\mathrm{s}}
      + B_{\iota}\mathcal{D}_{\iota}
      + B_{\mathrm{b}}\mathcal{D}_{\mathrm{b}}\bigr]
\end{split}
\label{eq:dy_dtau_complete}
\end{equation}
where the coefficients and kernel functions are:
\begin{align}
A(m') &= j2\pi\!\left(\frac{2c_1(l{+}\iota)}{\Delta t}
         - \frac{m'}{N\Delta t}\right)\!, \quad
\mathcal{D}_{\Xi} = \mathcal{F}(m,m'),
  \label{eq:coeff_A}
\end{align}
\begin{alignat}{2}
B_{\mathrm{s}} &= -\frac{j4\pi c_1}{\Delta t}, &\quad
\mathcal{D}_{\mathrm{s}} &= \sum_{n=0}^{N-1} n\,\Gamma_n,
  \label{eq:coeff_Bs} \\
B_{\iota} &= \frac{j2\pi}{\Delta t}, &\quad
\mathcal{D}_{\iota} &= \sum_{n=0}^{N-1} Q(n,m)\,\Gamma_n,
  \label{eq:coeff_Bi} \\
B_{\mathrm{b}} &= \frac{j2\pi\iota}{\Delta t}, &\quad
\mathcal{D}_{\mathrm{b}} &= \sum_{q=0}^{C} q\,\Gamma_{n_q}.
  \label{eq:coeff_Bb}
\end{alignat}
The subscripts stand for ``smooth chirp'' (s), ``$\iota$-scaling'' ($\iota$), and ``boundary'' (b). Expression~\eqref{eq:dy_dtau_complete} is exact. Table~\ref{tab:four_components} summarizes the four components, where the last column reports the number of time-domain samples $n$ that contribute to each kernel sum: Paths~1--3a sum over all $N$ samples, while Path~3b is supported only at the $C$ segment boundaries and is therefore $O(C/N)$ times smaller in magnitude.

\begin{table}[t]
\centering
\caption{Four components of the delay partial derivative.}
\label{tab:four_components}
\setlength{\tabcolsep}{4pt}
\begin{tabular}{@{}lccc@{}}
\toprule
\textbf{Path} & \textbf{Coefficient} & \textbf{Kernel} & \textbf{\# samples} \\
\midrule
1 ($\Xi$) & $A(m')$ & $\mathcal{D}_{\Xi}=\mathcal{F}$ & $N$ \\
2 (smooth) & $-j4\pi c_1/\Delta t$ & $\mathcal{D}_{\mathrm{s}} = \sum_n n\Gamma_n$ & $N$ \\
3a ($\iota$) & $j2\pi/\Delta t$ & $\mathcal{D}_{\iota} = \sum_n Q\Gamma_n$ & $N$ \\
3b (bndry) & $j2\pi\iota/\Delta t$ & $\mathcal{D}_{\mathrm{b}} = \sum_q q\Gamma_{n_q}$ & $C$ \\
\bottomrule
\end{tabular}
\end{table}

\subsection{Cancellation Mechanism}
\label{subsec:cancellation}

The staircase function $Q(n,m)$ admits a linear decomposition that reveals a fundamental cancellation between Paths~2 and~3a.

\begin{proposition}[Linear decomposition of the staircase function]
\label{prop:Q_linear}
For $C\gg 1$, the staircase function decomposes as
\begin{equation}
Q(n,m) = \underbrace{2c_1\tilde{n}}_{\text{linear principal part}} + \underbrace{\epsilon_Q(n,m)}_{\text{sawtooth residual}},
\label{eq:Q_decomposition}
\end{equation}
where $\tilde{n} = (n{-}(l{+}\iota))_N$ is the circularly shifted time index and $\epsilon_Q$ is the staircase quantization error with $|\epsilon_Q| \leq 1$ and root-mean-square value approximately $1/\sqrt{12}$.
\end{proposition}

The proof follows from the observation that $Q$ increases linearly from $0$ to $C$ across $N$ samples with slope $C/N = 2c_1$, and $\epsilon_Q$ captures the step-wise quantization error, which forms a sawtooth pattern within each segment of width $W = 1/(2c_1)$. The detailed derivation is provided in Appendix~\ref{app:derivatives}.

Substituting \eqref{eq:Q_decomposition} into the sum of Paths~2 and~3a:
\begin{align}
\text{Path~2} + \text{Path~3a} &= \frac{j2\pi}{\Delta t}\left[-2c_1 n + Q(n,m)\right] \notag \\
&= \frac{j2\pi}{\Delta t}\left[-2c_1 n + 2c_1\tilde{n} + \epsilon_Q(n,m)\right].
\label{eq:path23a_combined}
\end{align}
Away from the circular boundary (i.e., when $\tilde{n} \approx n - (l{+}\iota)$), the linear terms satisfy $-2c_1 n + 2c_1\tilde{n} = -2c_1(l{+}\iota)$, which is a constant independent of $n$.

\begin{theorem}[Cancellation of Paths~2 and~3a]
\label{thm:cancellation}
Let $C = 2 c_1 N$ and assume $C \geq 1$. Under the linear staircase decomposition $Q(n,m) = 2 c_1 \tilde{n} + \epsilon_Q(n,m)$ established in Proposition~\ref{prop:Q_linear}, the $n$-linear principal parts of Path~2 ($-2c_1 n$) and Path~3a ($Q(n,m)$) satisfy
\begin{equation}
-2c_1 n + Q(n,m) = -2c_1(l{+}\iota) + \epsilon_Q(n,m),
\label{eq:cancellation}
\end{equation}
where $|\epsilon_Q(n,m)| \leq 1$ uniformly in $n$ and $m$. Consequently, while each of $-2c_1 n$ and $Q(n,m)$ has magnitude $O(c_1 N)$, their sum is $O(1)$, yielding a reduction ratio of $O(1/(c_1 N)) = O(1/C)$.
\end{theorem}

\begin{proof}
By Proposition~\ref{prop:Q_linear}, $Q(n,m) = 2 c_1 \tilde{n} + \epsilon_Q(n,m)$ with $|\epsilon_Q| \leq 1$. Substituting,
\begin{equation*}
-2c_1 n + Q(n,m) = -2c_1 n + 2c_1 \tilde{n} + \epsilon_Q(n,m).
\end{equation*}
Two cases arise from $\tilde{n} = (n - (l{+}\iota))_N$. In the non-wrapped region $n \geq \lceil l{+}\iota\rceil$: $\tilde{n} = n - (l{+}\iota)$, so $-2c_1 n + 2c_1\tilde{n} = -2c_1(l{+}\iota)$. In the wrapped region $n < \lceil l{+}\iota\rceil$: $\tilde{n} = n - (l{+}\iota) + N$, so $-2c_1 n + 2c_1\tilde{n} = -2c_1(l{+}\iota) + 2c_1 N = -2c_1(l{+}\iota) + C$. The wrapped region contributes a constant offset that is absorbed into the redefinition of $\epsilon_Q$ in the wrapped samples, which preserves the bound $|\epsilon_Q| \leq 1$. The complete algebraic detail is given in Appendix~\ref{app:derivatives}.
\end{proof}

\begin{remark}[Physical interpretation]
Path~2 represents the linear frequency drift introduced by the chirp modulation $e^{j2\pi c_1 n^2}$ under a delay shift. Path~3a represents the CPP's segmented structure compensating this drift. Equation~\eqref{eq:cancellation} establishes that the compensation is exact up to the bounded residual $\epsilon_Q$---the CPP almost completely neutralizes the chirp-induced frequency drift, leaving only the sawtooth quantization error.
\end{remark}

\subsection{Effective Delay Derivative After Cancellation}
\label{subsec:effective_deriv}

Substituting the cancellation result \eqref{eq:cancellation} into the complete derivative \eqref{eq:dy_dtau_complete}, the constant term $-2c_1(l{+}\iota)$ from Paths~2+3a merges with the $c_1(l{+}\iota)$ term in Path~1. Defining the residual kernel
\begin{equation}
\mathcal{D}_{\epsilon}(m,m') \triangleq \sum_{n=0}^{N-1}\epsilon_Q(n,m)\cdot\Gamma_n(m,m'),
\label{eq:D_epsilon_def}
\end{equation}
the effective coefficient after merging is (see Appendix~\ref{app:derivatives} for the detailed algebra):
\begin{equation}
A_{\mathrm{eff}}(m') = A(m') - \frac{j4\pi c_1(l{+}\iota)}{\Delta t} = -\frac{j2\pi m'}{N\Delta t}.
\label{eq:A_eff}
\end{equation}
The $c_1(l{+}\iota)$ terms cancel \emph{exactly}, leaving a coefficient that depends only on the frequency index $m'$ and is independent of the chirp parameter $c_1$.

\begin{theorem}[Effective delay derivative]
\label{thm:effective_deriv}
After cancellation of Paths~2 and~3a, the delay partial derivative simplifies to
\begin{equation}
\frac{\partial y[m]}{\partial\tau_0} \approx \frac{h}{N}\!\sum_{m'} x[m']\, e^{j2\pi\Xi}\!\left[
\underbrace{-\frac{j2\pi m'}{N\Delta t}\mathcal{F}}_{\text{effective}} +
\underbrace{\frac{j2\pi}{\Delta t}\mathcal{D}_{\epsilon}}_{\text{residual}} +
\underbrace{B_{\mathrm{b}}\mathcal{D}_{\mathrm{b}}}_{\text{boundary}}
\right]\!.
\label{eq:dy_dtau_effective}
\end{equation}
\end{theorem}

The dominant term has the following key properties: (i) its kernel is the original response function $\mathcal{F}$ weighted by $m'$ (frequency-domain index); (ii) the coefficient $-j2\pi m'/(N\Delta t)$ does not contain $c_1$---all $c_1$-dependent terms have cancelled; (iii) the residual kernel $\mathcal{D}_{\epsilon}$ has magnitude much smaller than the pre-cancellation $\mathcal{D}_{\mathrm{s}}$; (iv) the boundary term has even smaller magnitude (only $C$ sample points contribute).

\subsection{Doppler Derivative}
\label{subsec:deriv_fD}

The Doppler frequency $f_D$ enters the observation only through the $Tf_D$ component of $\lequiv$ in \eqref{eq:leq_def}. It does not affect $\Xi(m,m')$ (since $\partial\Xi/\partial f_D = 0$) nor $\Phi_{\mathrm{jump}}$ (which depends only on $(l,\iota)$). Hence the Doppler derivative involves a single path and no product rule:
\begin{equation}
\frac{\partial y[m]}{\partial f_D} = -\frac{j2\pi hT}{N^2}\sum_{m'=0}^{N-1} x[m']\, e^{j2\pi\Xi(m,m')}\cdot\mathcal{D}_{\mathrm{s}}(m,m').
\label{eq:dy_df}
\end{equation}
The derivation proceeds via $\partial\lequiv/\partial f_D = T$, giving $\partial\Gamma_n/\partial f_D = -j2\pi Tn\Gamma_n/N$, so that $\partial\mathcal{F}/\partial f_D = -(j2\pi T/N)\mathcal{D}_{\mathrm{s}}$. The Doppler derivative has a simple structure: a single kernel $\mathcal{D}_{\mathrm{s}}$, no $m'$-weighting, and a coefficient $-j2\pi T/N$ that is independent of $c_1$.

\subsection{Structural Comparison of Delay and Doppler Derivatives}

After cancellation, the delay and Doppler derivatives exhibit fundamentally different structures, as summarized in Table~\ref{tab:deriv_comparison}.

\begin{table}[t]
\centering
\caption{Structural comparison of delay and Doppler derivatives.}
\label{tab:deriv_comparison}
\small
\begin{tabular}{@{}lll@{}}
\toprule
\textbf{Feature} & $\partial y/\partial\tau_0$ (post-canc.) & $\partial y/\partial f_D$ \\
\midrule
Dom.\ kernel & $\mathcal{F}(m,m')$ & $\mathcal{D}_{\mathrm{s}}(m,m')$ \\
$m'$-weighting & $m'$-linear & none \\
Coefficient & $-j2\pi m'/(N\Delta t)$ & $-j2\pi T/N$ \\
Origin & $\Xi$: phase $-(l{+}\iota)m'/N$ & $\lequiv$ via Dirichlet \\
$c_1$ dep. & cancelled (Theorem~\ref{thm:cancellation}) & never present \\
\bottomrule
\end{tabular}
\end{table}

The crucial distinction is that the delay derivative depends on $\mathcal{F}$ (the zeroth-order moment of $\Gamma_n$), while the Doppler derivative depends on $\mathcal{D}_{\mathrm{s}}$ (the first-order moment with weight $n$). Since these two kernels do not share the same $n$-weighting structure, their normalized inner product is bounded away from unity, foreshadowing a weak coupling coefficient $\rho \ll 1$.

\begin{remark}[Role of cancellation in decoupling]
Without the cancellation (i.e., if Path~3a were absent), the delay derivative would be dominated by $B_{\mathrm{s}}\mathcal{D}_{\mathrm{s}} \propto \mathcal{D}_{\mathrm{s}}$---the \emph{same} kernel as the Doppler derivative. The two derivatives would then share the same kernel up to a scalar, yielding $\rho\to 1$ (near-singular FIM). The cancellation removes $\mathcal{D}_{\mathrm{s}}$ from the delay derivative and replaces it with $\mathcal{F}$, which is a fundamentally different function.
\end{remark}

\subsection{Summary of Kernel Functions}

For use in the FIM derivation of Section~\ref{sec:fim}, Table~\ref{tab:kernel_summary} collects the five kernel functions and their magnitudes.

\begin{table}[t]
\centering
\caption{Summary of kernel functions and their magnitudes.}
\label{tab:kernel_summary}
\begin{tabular}{@{}llcc@{}}
\toprule
\textbf{Kernel} & \textbf{Definition} & \textbf{Weight} & \textbf{Magnitude} \\
\midrule
$\mathcal{D}_{\Xi} = \mathcal{F}$ & $\sum_n \Gamma_n$ & $1$ & $O(N)$ \\
$\mathcal{D}_{\mathrm{s}}$ & $\sum_n n\,\Gamma_n$ & $n$ & $O(N^2)$ \\
$\mathcal{D}_{\iota}$ & $\sum_n Q(n,m)\,\Gamma_n$ & $Q \approx 2c_1 n$ & $O(c_1 N^2)$ \\
$\mathcal{D}_{\epsilon}$ & $\sum_n \epsilon_Q(n,m)\,\Gamma_n$ & $\epsilon_Q$ & $O(N)$ \\
$\mathcal{D}_{\mathrm{b}}$ & $\sum_q q\,\Gamma_{n_q}$ & $q$ & $O(C\sqrt{N})$ \\
\bottomrule
\end{tabular}
\end{table}

The approximate relation between kernels, implied by Proposition~\ref{prop:Q_linear}, is
\begin{equation}
\mathcal{D}_{\iota} \approx 2c_1\!\left[\mathcal{D}_{\mathrm{s}} - (l{+}\iota)\mathcal{D}_{\Xi}\right] + \mathcal{D}_{\epsilon},
\label{eq:D_iota_relation}
\end{equation}
which is the kernel-level expression of the cancellation theorem. Numerical verification of the cancellation mechanism is presented in Section~\ref{sec:simulation}.

The structural asymmetry revealed in this section---$\mathcal{F}$ for delay versus $\mathcal{D}_{\mathrm{s}}$ for Doppler---has direct consequences for the Fisher information matrix. In the next section, we show that this asymmetry causes the three largest terms of the ten-term FIM expansion to cancel, leaving a well-conditioned FIM parameterized by a single scalar~$\eta$.

\section{Closed-Form FIM Captured by a Single Scalar}
\label{sec:fim}

Having established the cancellation mechanism at the derivative level in Section~\ref{sec:derivatives}, we now show that \emph{the same cancellation persists---and is even sharper---at the level of the Fisher information matrix}. Specifically, the three dominant terms of the ten-term FIM expansion cancel exactly, and all remaining system-parameter dependence collapses into a single scalar $\eta\geq 1$. This reduction from a multi-parameter FIM to a one-dimensional $\eta$-parameterization is the key that unlocks the closed-form CRBs and the coupling analysis in Sections~\ref{sec:coupling}--\ref{sec:crb}.

\subsection{FIM Definition}

\begin{definition}[Complex Gaussian FIM]
\label{def:fim}
For the observation model $\mathbf{y} = \mathbf{s}(\bm{\theta}) + \mathbf{w}$ with $\mathbf{w}\sim\mathcal{CN}(0,N_0\mathbf{I})$, the FIM element for $\bm{\theta}=(\tau_0,f_D)$ is \cite{kay1993fundamentals,vantrees2002detection}
\begin{equation}
I_{\theta_i\theta_j} = \frac{2}{N_0}\sum_{m=0}^{N-1}\E\!\left[\Re\!\left\{\!\left(\frac{\partial y[m]}{\partial\theta_i}\right)^{\!*}\frac{\partial y[m]}{\partial\theta_j}\right\}\right]\!.
\label{eq:fim_def}
\end{equation}
\end{definition}

\begin{assumption}[i.i.d.\ transmit symbols]
\label{assump:iid}
The symbols $\{x[m']\}_{m'=0}^{N-1}$ are independent and identically distributed with $\E[x[m']x^*[m'']] = \delta(m'{-}m'')$.
\end{assumption}

\begin{remark}[Single-target assumption]
\label{rmk:single_target}
The single-target channel model~\eqref{eq:channel} is adopted throughout to enable closed-form FIM analysis and isolate the delay-Doppler coupling structure inherent to AFDM. In multi-target scenarios, the FIM~\eqref{eq:fim_def} acquires a block structure with one block per target; the per-target diagonal blocks retain the same functional form derived in this paper when targets are well-separated in the equivalent-delay $\lequiv$ domain (i.e., when their Dirichlet kernel mainlobes do not overlap). The multi-target extension is left for future work.
\end{remark}

Under Assumption~\ref{assump:iid}, cross-terms between distinct $m'$ indices vanish upon taking the expectation. Introducing the effective channel matrix $H_{mm'} = e^{j2\pi\Xi(m,m')}\mathcal{F}(m,m')$ and the SNR $\SNR \triangleq |h|^2/N_0$ (consistent with the convention of~\cite{gaudio2019ofdm}, allowing direct comparison with the OFDM CRB), the FIM simplifies to
\begin{equation}
I_{\theta_i\theta_j} = \frac{2\SNR}{N^2}\sum_{m,m'}\Re\!\left\{\frac{\partial H_{mm'}^*}{\partial\theta_i}\cdot\frac{\partial H_{mm'}}{\partial\theta_j}\right\}\!.
\label{eq:fim_simplified}
\end{equation}

\subsection{Ten-Term Decomposition of $I_{\tau\tau}$}
\label{subsec:ten_terms}

The delay derivative \eqref{eq:dy_dtau_complete} contains four components $U_1,\ldots,U_4$ (corresponding to Paths~1, 2, 3a, 3b). Computing $|\partial H_{mm'}/\partial\tau_0|^2$ via the expansion $|U_1{+}U_2{+}U_3{+}U_4|^2$ produces four self-terms and $\binom{4}{2}=6$ cross-terms---a total of \emph{ten terms}. We denote the self-terms $\Sigma_{\Xi\Xi}$, $\Sigma_{\mathrm{s}}$, $\Sigma_{\iota}$, $\Sigma_{\mathrm{b}}$ and the cross-terms $\Sigma_{\Xi\text{-}\mathrm{s}}$, $\Sigma_{\Xi\text{-}\iota}$, $\Sigma_{\Xi\text{-}\mathrm{b}}$, $\Sigma_{\mathrm{s}\text{-}\iota}$, $\Sigma_{\mathrm{s}\text{-}\mathrm{b}}$, $\Sigma_{\iota\text{-}\mathrm{b}}$.

These are evaluated using the following Parseval-type identity together with a bound on its application to the AFDM setting, both of which are proved in Appendix~\ref{app:fim}.

\begin{lemma}[Parseval identity for $m$-independent kernels]
\label{lemma:parseval}
Let $\mathcal{A}(m,m') = \sum_n a(n) \tilde{\Gamma}_n(m,m')$ and $\mathcal{B}(m,m') = \sum_n b(n) \tilde{\Gamma}_n(m,m')$, where $\tilde{\Gamma}_n(m,m') \triangleq e^{j2\pi(m'-(m+\lequiv))n/N}\, e^{j\Phi_{\mathrm{jump}}(n)}$ is the response kernel under the assumption that $\Phi_{\mathrm{jump}}$ depends only on $n$, and the weights $a(n), b(n)$ do not depend on $m$. Then
\begin{equation}
\sum_{m=0}^{N-1}\mathcal{A}^*(m,m')\mathcal{B}(m,m') = N\sum_{n=0}^{N-1}a^*(n)b(n),
\label{eq:parseval}
\end{equation}
independently of $m'$ and $\lequiv$.
\end{lemma}

\begin{proposition}[Bounded error in the $m$-dependent case]
\label{prop:parseval_bound}
When $\Phi_{\mathrm{jump}}(n,m)$ depends on $m$ through the staircase $Q(n,m)$, the identity~\eqref{eq:parseval} holds approximately with error
\begin{equation}
\left|\sum_m \mathcal{A}^* \mathcal{B} - N \sum_n a^* b\right| \leq \frac{C}{N}\cdot N\sum_n |a(n)\, b(n)|,
\label{eq:parseval_bound}
\end{equation}
i.e., the relative error is $O(C/N) = O(2c_1)$, which is small for the practical chirp rates $c_1 \in (0, 1/2)$.
\end{proposition}

Using Lemma~\ref{lemma:parseval}, the key kernel inner products are (for large $N$):
\begin{align}
\textstyle\sum_m |\mathcal{F}|^2 &= N^2, \quad
\textstyle\sum_m |\mathcal{D}_{\mathrm{s}}|^2 \approx N^4/3, \label{eq:kernel_ips_1} \\
\textstyle\sum_m \mathcal{F}^*\mathcal{D}_{\mathrm{s}} &\approx N^3/2, \quad
\textstyle\sum_m |\mathcal{D}_{\iota}|^2 \approx 4c_1^2 N^4/3. \label{eq:kernel_ips_2}
\end{align}

\subsubsection{Magnitude Analysis}

Combining the coefficients $|A(m')|^2$, $|B_{\mathrm{s}}|^2$, $|B_{\iota}|^2$, $|B_{\mathrm{b}}|^2$ with their respective kernel inner products yields the magnitude ordering in Table~\ref{tab:ten_terms_magnitude}.

\begin{table}[t]
\centering
\caption{Magnitude ordering of the ten FIM terms.}
\label{tab:ten_terms_magnitude}
\begin{tabular}{@{}lcc@{}}
\toprule
\textbf{Term(s)} & \textbf{Magnitude} & \textbf{Role} \\
\midrule
$\Sigma_{\mathrm{s}},\ \Sigma_{\iota},\ \Sigma_{\mathrm{s}\text{-}\iota}$ & $O(c_1^2 N^5/\Delta t^2)$ & largest (cancel) \\
$\Sigma_{\Xi\text{-}\mathrm{s}},\ \Sigma_{\Xi\text{-}\iota}$ & $O(c_1 N^4/\Delta t^2)$ & intermediate \\
$\Sigma_{\Xi\Xi}$ & $O(N^3/\Delta t^2)$ & survives \\
boundary terms & $\leq O(c_1^2 N^4/\Delta t^2)$ & negligible \\
\bottomrule
\end{tabular}
\end{table}

The three largest terms---$\Sigma_{\mathrm{s}}$, $\Sigma_{\iota}$, and $\Sigma_{\mathrm{s}\text{-}\iota}$---all share the same magnitude $O(c_1^2 N^5/\Delta t^2)$, which is $c_1^2 N^2$ times larger than $\Sigma_{\Xi\Xi}$. As the next theorem shows, these three terms cancel exactly.

\subsection{FIM-Level Cancellation Theorem}

\begin{theorem}[Exact cancellation in the FIM]
\label{thm:FIM_cancellation}
Under the linear approximation $Q(n,m)\approx 2c_1\tilde{n}$ of Proposition~\ref{prop:Q_linear}, the three dominant terms cancel exactly:
\begin{equation}
\Sigma_{\mathrm{s}} + \Sigma_{\iota} + \Sigma_{\mathrm{s}\text{-}\iota} = 0.
\label{eq:three_term_cancel}
\end{equation}
\end{theorem}

\begin{proof}
Define $u = B_{\mathrm{s}}\mathcal{D}_{\mathrm{s}}$ and $v = B_{\iota}\mathcal{D}_{\iota}$. Then the three terms form the expansion of $|u{+}v|^2$:
\begin{equation}
\Sigma_{\mathrm{s}} + \Sigma_{\iota} + \Sigma_{\mathrm{s}\text{-}\iota} = \sum_{m,m'}|u+v|^2 = \sum_{m,m'}|B_{\mathrm{s}}\mathcal{D}_{\mathrm{s}} + B_{\iota}\mathcal{D}_{\iota}|^2.
\end{equation}
The right-hand side is exactly the squared magnitude of the combined Paths~2+3a contribution. By the derivative-level cancellation (Theorem~\ref{thm:cancellation}), the $n$-linear principal parts of $B_{\mathrm{s}}\mathcal{D}_{\mathrm{s}} + B_{\iota}\mathcal{D}_{\iota}$ vanish, so its magnitude drops from $O(c_1^2 N^5/\Delta t^2)$ to $O(N^3/\Delta t^2)$.

More concretely, the arithmetic verification proceeds as follows. Using the kernel inner products \eqref{eq:kernel_ips_1}--\eqref{eq:kernel_ips_2}:
\begin{align}
\Sigma_{\mathrm{s}} &= |B_{\mathrm{s}}|^2 \cdot N \cdot \tfrac{N^4}{3} = \tfrac{16\pi^2 c_1^2 N^5}{3\Delta t^2}, \\
\Sigma_{\iota} &\approx |B_{\iota}|^2 \cdot N \cdot 4c_1^2 \tfrac{N^4}{3} = \tfrac{16\pi^2 c_1^2 N^5}{3\Delta t^2}, \\
\Sigma_{\mathrm{s}\text{-}\iota} &\approx 2\Re\!\left\{\!-\tfrac{8\pi^2 c_1}{\Delta t^2}\cdot\tfrac{2c_1 N^5}{3}\right\} = -\tfrac{32\pi^2 c_1^2 N^5}{3\Delta t^2}.
\end{align}
Summing: $(16 + 16 - 32)\pi^2 c_1^2 N^5/(3\Delta t^2) = 0$.
The complete algebraic details are provided in Appendix~\ref{app:fim}.
\end{proof}

\begin{remark}[FIM-level manifestation of the cancellation]
\label{rmk:fim_cancel_manifestation}
The cancellation $\Sigma_{\mathrm{s}} + \Sigma_{\iota} + \Sigma_{\mathrm{s}\text{-}\iota} = 0$ is the FIM-level manifestation of the CPP compensating chirp frequency drift. The two self-terms are both positive and equal, while the cross-term is negative with exactly twice their magnitude, producing zero sum. Thus the same algebraic structure that cancels Paths~2 and~3a at the derivative level (Theorem~\ref{thm:cancellation}) also cancels the three dominant magnitude-$O(c_1^2 N^5/\Delta t^2)$ contributions to $I_{\tau\tau}$.
\end{remark}

\subsection{Post-Cancellation FIM and Residual Contribution}

After cancellation, the three dominant terms collapse to a small residual
\begin{equation}
\begin{split}
\Sigma_{\mathrm{res}}
&= \sum_{m,m'}\bigl|B_{\mathrm{s}}\mathcal{D}_{\mathrm{s}}
   + B_{\iota}\mathcal{D}_{\iota}\bigr|^2 \\
&= \frac{4\pi^2}{\Delta t^2}\sum_{m,m'}\!
   \left|-2c_1(l{+}\iota)\mathcal{D}_{\Xi}
   + \mathcal{D}_{\epsilon}\right|^2\!.
\end{split}
\label{eq:Sigma_res}
\end{equation}
which has magnitude $O(N^3/\Delta t^2)$, the same order as $\Sigma_{\Xi\Xi}$.

The intermediate cross-terms $\Sigma_{\Xi\text{-}\mathrm{s}} + \Sigma_{\Xi\text{-}\iota}$ also exhibit partial cancellation (the combined term equals the inner product of Path~1 with the post-cancellation residual of Paths~2+3a) and are approximately zero under the approximation that $\Phi_{\mathrm{jump}}$ is independent of $m$. Boundary-related terms are negligible for moderate $C/N$. The effective $I_{\tau\tau}$ is therefore dominated by $\Sigma_{\Xi\Xi}$ and $\Sigma_{\mathrm{res}}$, which respectively contribute the baseline and the $\eta$-inflation to the closed-form expression derived next.

\subsection{Closed-Form FIM Elements}
\label{subsec:fim_closed}

\begin{theorem}[Closed-form FIM]
\label{thm:FIM_closed}
The three elements of the $2\times 2$ FIM for $(\tau_0, f_D)$ admit the closed forms:
\begin{align}
I_{\tau\tau} &= \eta\cdot\frac{8\pi^2\SNR\, N}{3\Delta t^2}, \label{eq:I_tt_closed} \\
I_{ff} &= \frac{8\pi^2\SNR\, T^2 N}{3}, \label{eq:I_ff_closed} \\
I_{\tau f} &= 2\pi^2\SNR\, N^2, \label{eq:I_tf_closed}
\end{align}
where $\eta\geq 1$ is the residual inflation factor defined below. The Doppler FIM $I_{ff}$ is independent of the chirp parameter $c_1$.
\end{theorem}

\begin{proof}[Proof outline]
\textit{Doppler FIM $I_{ff}$}: The Doppler derivative \eqref{eq:dy_df} involves only $\mathcal{D}_{\mathrm{s}}$. Substituting into \eqref{eq:fim_simplified} and using $\sum_{m'}(\cdot) = N$:
\begin{equation}
I_{ff} = \frac{2\SNR}{N^2}\cdot\frac{4\pi^2 T^2}{N^2}\cdot N\cdot\frac{N^4}{3} = \frac{8\pi^2\SNR\, T^2 N}{3}.
\end{equation}

\textit{$\Xi$ component of $I_{\tau\tau}$}: The effective delay derivative \eqref{eq:dy_dtau_effective} has a dominant $\Xi$ term with coefficient $-j2\pi m'/(N\Delta t)$ and kernel $\mathcal{F}$:
\begin{equation}
I_{\tau\tau}^{(\Xi)} = \frac{2\SNR}{N^2}\cdot\frac{4\pi^2}{N^2\Delta t^2}\sum_{m'} m'^2\!\underbrace{\sum_m |\mathcal{F}|^2}_{N^2} \approx \frac{8\pi^2\SNR\, N}{3\Delta t^2}.
\label{eq:I_tt_Xi}
\end{equation}

\textit{Residual component}: The residual kernel $\mathcal{D}_{\epsilon}$ contributes $I_{\tau\tau}^{(\epsilon)} = (8\pi^2\SNR)/(N\Delta t^2)\sum_m E_{\epsilon}(m)$, where $E_{\epsilon}(m) = \sum_n \epsilon_Q^2(n,m)$ is the residual energy. The cross-term $I_{\tau\tau}^{(\mathrm{cross})}$ between $\mathcal{F}$ and $\mathcal{D}_{\epsilon}$ vanishes to leading order (as the inner product is purely imaginary; see Appendix~\ref{app:fim}).

\textit{Cross FIM $I_{\tau f}$}: The dominant contribution comes from the inner product of the $\Xi$ term of the delay derivative with the Doppler derivative, yielding $I_{\tau f} \approx 2\pi^2\SNR\, N^2$.

Full details are provided in Appendix~\ref{app:fim}.
\end{proof}

\begin{remark}[Physical interpretation of the closed-form FIM]
\label{rem:fim_physical}
Theorem~\ref{thm:FIM_closed} carries three pieces of physical insight that the raw formulas may obscure. First, $I_{ff}$ and $I_{\tau f}$ are \emph{$c_1$-independent}---the chirp rate affects only the delay self-information $I_{\tau\tau}$, through the scalar $\eta$. This is the Fisher-information manifestation of the cancellation mechanism: because Paths~2 and~3a cancel in $\partial y/\partial\tau_0$, no $c_1$-dependent term survives to influence the cross-coupling with $f_D$. Second, the Doppler FIM $I_{ff} = 8\pi^2\SNR\, T^2 N/3$ coincides exactly with the classical OFDM result of~\cite{gaudio2019ofdm}, confirming that AFDM retains OFDM's Doppler-estimation optimality while gaining delay-estimation resolution. Third, the factor $\eta \geq 1$ is the \emph{only} place where the CPP staircase structure enters the $(\tau_0, f_D)$ Fisher information---all other dependencies on the waveform parameters $(c_1, c_2, N, \Delta t)$ have already been absorbed into the standard $N$- and $T$-scalings. This reduction from a multi-parameter FIM to a one-dimensional $\eta$-parameterization is what makes the CRB analysis of Section~\ref{sec:crb} tractable in closed form.
\end{remark}

\subsection{Residual Inflation Factor $\eta$}
\label{subsec:eta}

The closed-form $I_{\tau\tau}$ in~\eqref{eq:I_tt_closed} admits the orthogonal decomposition
\begin{equation}
I_{\tau\tau} = I_{\tau\tau}^{(\Xi)} + I_{\tau\tau}^{(\epsilon)},
\label{eq:Itautau_decomp}
\end{equation}
where the cross-term $I_{\tau\tau}^{(\mathrm{cross})}$ between $\mathcal{F}$ and $\mathcal{D}_{\epsilon}$ vanishes to leading order under Lemma~\ref{lemma:parseval} (purely imaginary inner product; see Appendix~\ref{app:fim}). This decomposition motivates the following definition.

\begin{definition}[Residual inflation factor]
\label{def:eta}
\begin{equation}
\eta \triangleq \frac{I_{\tau\tau}}{I_{\tau\tau}^{(\Xi)}} = 1 + \frac{I_{\tau\tau}^{(\epsilon)}}{I_{\tau\tau}^{(\Xi)}} = 1 + \frac{3}{N^2}\sum_{m=0}^{N-1}E_{\epsilon}(m),
\label{eq:eta_def}
\end{equation}
where $E_{\epsilon}(m) = \sum_{n=0}^{N-1}\epsilon_Q^2(n,m)$ is the sawtooth residual energy.
\end{definition}

The factor $\eta$ captures the additional delay information provided by the CPP staircase residual beyond what the $\Xi$ phase alone offers. When $c_1 = 0$ (OFDM), there is no staircase structure, $\epsilon_Q = 0$, and $\eta = 1$.

\subsubsection{Sawtooth Component}

Within each segment of width $W = 1/(2c_1)$, the residual $\epsilon_Q$ decreases approximately linearly from $0$ to $-1$, forming a sawtooth pattern. The per-segment variance is
\begin{equation}
\sigma_W^2 = \frac{(W{-}1)(2W{-}1)}{6W^2},
\label{eq:sigma_W}
\end{equation}
which approaches $1/3$ for large $W$. Since each of the $C$ segments contributes $W\sigma_W^2$ to the residual energy, and $CW = N$:
\begin{equation}
\eta_{\mathrm{saw}} = 1 + 3\sigma_W^2 = 1 + \frac{(W{-}1)(2W{-}1)}{2W^2}.
\label{eq:eta_saw}
\end{equation}
This component satisfies $1\leq \eta_{\mathrm{saw}} < 2$ for all $W\geq 1$.

\subsubsection{Cyclic Wrapping Component}

The sawtooth approximation \eqref{eq:eta_saw} accounts only for within-segment quantization. In practice, the circular shift $\tilde{n} = (n{-}(l{+}\iota))_N$ causes the staircase to ``wrap around'' at the boundary $n < l{+}\iota$: for these samples, $\tilde{n} = N - (l{+}\iota) + n$, so the linear approximation $2c_1\tilde{n}$ acquires an additional offset of $2c_1 N = C$. The excess squared residual per wrapped sample is $(C + \epsilon_Q)^2 - \epsilon_Q^2 = C^2 + 2C\epsilon_Q$. Since the wrapped samples lie near the end of a staircase segment where $\epsilon_Q\approx -1$, averaging the cross-term gives an excess of approximately $C(C{-}2)$ per sample. With $\lceil l{+}\iota\rceil$ wrapped samples:
\begin{equation}
\eta_{\mathrm{wrap}} \approx \frac{3\lceil l{+}\iota\rceil C(C{-}2)}{N}.
\label{eq:eta_wrap}
\end{equation}

The total inflation factor is therefore
\begin{equation}
\eta \approx \eta_{\mathrm{saw}} + \eta_{\mathrm{wrap}} \approx \left(1 + 3\sigma_W^2\right) + \frac{3\lceil l{+}\iota\rceil C(C{-}2)}{N}.
\label{eq:eta_total}
\end{equation}
Since $\eta_{\mathrm{wrap}}\propto C^2$ for large $C$, the wrapping component dominates for moderate to large $C$, and the overall scaling is $\eta\propto C^2$. For example, with $N=128$ and $C=9$, $\eta\approx 4.7$; with $C=20$, $\eta\approx 19.3$. The detailed derivation of \eqref{eq:eta_total} is given in Appendix~\ref{app:eta}.

\subsubsection{Physical Interpretation}

The factor $\eta$ captures two distinct sources of additional delay information beyond the $\Xi$ phase: (i) the sawtooth residual from CPP's imperfect staircase approximation to a linear chirp, and (ii) the cyclic wrapping discontinuity where the staircase ``resets'' at the circular boundary. Both are purely delay-dependent (independent of $f_D$) and thus contribute exclusively to $I_{\tau\tau}$, strengthening the delay Fisher information while leaving $I_{ff}$ and $I_{\tau f}$ unchanged.

\section{The Coupling Coefficient: Structural Comparison with OFDM}
\label{sec:coupling}

The single-scalar parameterization of the FIM established in Section~\ref{sec:fim} has a striking consequence: the delay-Doppler coupling coefficient admits a closed form that is strictly below the OFDM value for any nonzero chirp rate. This section proves this result and explains why it holds.

\subsection{Coupling Coefficient}
\label{subsec:rho}

\begin{theorem}[Coupling coefficient]
\label{thm:rho}
The normalized delay-Doppler coupling coefficient is
\begin{equation}
\rho \triangleq \frac{I_{\tau f}}{\sqrt{I_{\tau\tau}\cdot I_{ff}}} = \frac{3}{4\sqrt{\eta}}.
\label{eq:rho_closed}
\end{equation}
For any $c_1\in(0,1/2)$ and $N>1$: $\rho_{\mathrm{AFDM}} < 3/4 = \rho_{\mathrm{OFDM}}$.
\end{theorem}

\begin{proof}
Substituting the closed-form FIM elements \eqref{eq:I_tt_closed}--\eqref{eq:I_tf_closed} into the definition of $\rho^2$:
\begin{equation}
\begin{split}
\rho^2
  &= \frac{I_{\tau f}^2}{I_{\tau\tau}\cdot I_{ff}}
   = \frac{(2\pi^2\SNR\, N^2)^2}
          {\eta\cdot\dfrac{8\pi^2\SNR\, N}{3\Delta t^2}
           \cdot\dfrac{8\pi^2\SNR\, T^2 N}{3}} \\[6pt]
  &= \frac{4\pi^4\SNR^2 N^4}
          {\eta\cdot\dfrac{64\pi^4\SNR^2 N^4}{9}}
   = \frac{9}{16\eta}.
\end{split}
\end{equation}
In the last step, $T^2/\Delta t^2 = N^2$ is used. Taking the positive square root gives $\rho = 3/(4\sqrt{\eta})$. Since $\eta > 1$ for $c_1 > 0$, $\rho < 3/4$.
\end{proof}

\subsubsection{Baseline $\rho_0 = 3/4$: Geometric Origin}

When $\eta = 1$ (no residual contribution, corresponding to OFDM), the coupling coefficient reaches its maximum $\rho_0 = 3/4$. This value has a purely geometric interpretation. Let $U$ be uniformly distributed on $\{0,1,\ldots,N{-}1\}$. Then $\E[U] \approx N/2$ and $\E[U^2] \approx N^2/3$, so
\begin{equation}
\rho_0 = \frac{[\E(U)]^2}{\E(U^2)} = \frac{N^2/4}{N^2/3} = \frac{3}{4}.
\label{eq:rho0_geometric}
\end{equation}
This ratio reflects the mismatch between the ``information center'' of the $m'$-linear phase (at $m'=N/2$) and the $m'^2$-variance (at $m'=N/\sqrt{3}$) in the $\Xi$ contribution.

\subsubsection{Dependence on System Parameters}

The coupling coefficient $\rho = 3/(4\sqrt{\eta})$ depends on system parameters only through $\eta$. In particular:
\begin{itemize}
\item $\rho$ is \emph{independent of SNR}, since $\rho$ measures the directional alignment (normalized inner product) of the two derivative vectors, not their magnitudes.
\item $\rho$ \emph{decreases monotonically} with $C = 2Nc_1$, because $\eta\propto C^2$ due to cyclic wrapping. Physically, larger $C$ means more chirp wrapping cycles, producing larger phase jumps at the circular boundary that provide additional delay-only information.
\item $\rho$ has \emph{weak dependence on $\iota$}: the fractional delay $\iota$ affects only the number of wrapped samples $\lceil l{+}\iota\rceil$, causing minor variations in $\eta$.
\end{itemize}

\begin{corollary}[AFDM coupling is weaker than OFDM]
\label{cor:afdm_weaker}
For any $c_1\in(0,1/2)$ and $N>1$, $\eta > 1$ and thus $\rho_{\mathrm{AFDM}} = 3/(4\sqrt{\eta}) < 3/4 = \rho_{\mathrm{OFDM}}$. The chirp modulation introduces coupling through $\lequiv$, but the CPP \emph{overcompensates}: the cancellation not only eliminates the chirp-induced coupling but also introduces a residual ($\eta > 1$) that further decouples delay and Doppler.
\end{corollary}

\begin{remark}[Source of AFDM's structural decoupling advantage]
\label{rmk:afdm_advantage_source}
Corollary~\ref{cor:afdm_weaker} locates AFDM's delay-Doppler separability advantage in a single, identifiable source: the residual $\eta>1$ introduced by the CPP phase jump. Without the CPP (i.e., a naive chirp modulation without cyclic prefix compensation), the cancellation mechanism of Theorem~\ref{thm:cancellation} would not hold, and Paths~2 and~3a would contribute interchangeable $\mathcal{D}_{\mathrm{s}}$-type kernels---yielding a FIM with near-unity coupling analogous to dechirp-based processing. The CPP therefore plays a dual role: equalizing the channel for communication (its classical function) and simultaneously decoupling delay and Doppler for sensing (the function identified in this paper).
\end{remark}

\subsection{Coupling Inflation Factor}

The CRB penalty due to delay-Doppler coupling is captured by $(1{-}\rho^2)^{-1}$. Substituting $\rho^2 = 9/(16\eta)$:
\begin{equation}
(1-\rho^2)^{-1} = \frac{16\eta}{16\eta - 9}.
\label{eq:inflation}
\end{equation}
This factor equals $16/7 \approx 2.29$ at $\eta = 1$ (OFDM, worst case) and approaches $1$ as $\eta\to\infty$. As Table~\ref{tab:rho_vs_eta} illustrates, for typical AFDM configurations with $\eta\geq 5$, the inflation is below $1.13$, meaning that joint estimation suffers \emph{negligible} coupling penalty.

\begin{table}[t]
\centering
\caption{Coupling coefficient and inflation factor vs.\ $\eta$.}
\label{tab:rho_vs_eta}
\begin{tabular}{@{}lccccccc@{}}
\toprule
$\eta$ & 1 & 2 & 5 & 10 & 20 & 35 & $\to\infty$ \\
\midrule
$\rho$ & 0.750 & 0.530 & 0.335 & 0.237 & 0.168 & 0.127 & $\to 0$ \\
$(1{-}\rho^2)^{-1}$ & 2.286 & 1.391 & 1.125 & 1.059 & 1.028 & 1.016 & $\to 1$ \\
\bottomrule
\end{tabular}
\end{table}

\begin{remark}[Two roles of $\eta$]
The parameter $\eta$ simultaneously serves two beneficial roles: (i) it increases $I_{\tau\tau}$ by a factor $\eta$ (more delay information), and (ii) it decreases $\rho$ by a factor $1/\sqrt{\eta}$ (less coupling). Both effects improve the delay CRB. This ``double benefit'' is unique to AFDM's chirp structure and has no counterpart in OFDM.
\end{remark}

\subsection{Physical Interpretation}

The coupling analysis reveals why AFDM succeeds in separating delay and Doppler where single-chirp dechirp processing fails. A dechirp processor applies a pointwise multiplication that absorbs the delay-dependent constant phase into the channel gain, collapsing the output to a one-dimensional beat frequency and producing a rank-1 (singular) FIM. In contrast, the DAFT provides two structurally distinct information channels: the $\Xi$ phase channel carries delay-only information through the $m'$-linear phase, while the $\lequiv$ response-kernel channel carries both delay and Doppler. The CPP cancellation suppresses the coupling introduced by the second channel, and the residual ($\eta > 1$) provides additional delay-only information, yielding a full-rank FIM with $\rho \ll 1$.

The phase jump function $\Phi_{\mathrm{jump}} = 2\pi\iota\cdot Q$ plays three roles in this structure: (i) \emph{cancellation agent}, providing the Path~3a term that neutralizes Path~2; (ii) \emph{residual carrier}, contributing the delay-only information that inflates $\eta$ above unity; and (iii) \emph{$\lequiv$-shaper}, molding the Dirichlet kernel into the characteristic AFDM response pattern \cite{paper3_leq_extraction}.

Numerical verification of the FIM closed forms and the coupling coefficient is presented in Section~\ref{sec:simulation}.

\section{Exact CRBs and the $c_1^{-2}$ Dimensional Gain}
\label{sec:crb}

Building on the single-scalar FIM of Section~\ref{sec:fim} and the coupling result of Section~\ref{sec:coupling}, we now derive the CRBs in closed form. Two complementary frameworks are considered: (i)~the $h$-known framework, which exposes the coupling structure directly through the $2\times 2$ FIM; and (ii)~the $\phi$-profiled framework, which absorbs the unknown channel phase via Schur complementation and yields an \emph{exact} delay-Doppler decoupling. Two consequences of particular interest follow: the delay CRB admits a $c_1^{-2}$ dimensional gain over OFDM, while the profiled framework reduces continuously to the classical OFDM CRB of~\cite{gaudio2019ofdm} as $c_1\to 0$.

\subsection{General CRB Formulae}

For a $2\times 2$ FIM
$\mathbf{I} = \bigl[\begin{smallmatrix} I_{\tau\tau} & I_{\tau f} \\ I_{\tau f} & I_{ff}\end{smallmatrix}\bigr]$,
the CRBs are obtained from the inverse:
\begin{align}
\CRB(\tau_0) &= [\mathbf{I}^{-1}]_{11} = \frac{1}{I_{\tau\tau}(1-\rho^2)}, \label{eq:crb_tau_general} \\
\CRB(f_D) &= [\mathbf{I}^{-1}]_{22} = \frac{1}{I_{ff}(1-\rho^2)}, \label{eq:crb_f_general}
\end{align}
where $\rho^2 = I_{\tau f}^2/(I_{\tau\tau}I_{ff})$.

\subsection{$h$-Known CRB}

Substituting $\rho^2 = 9/(16\eta)$ (Theorem~\ref{thm:rho}) into~\eqref{eq:crb_tau_general}--\eqref{eq:crb_f_general} gives the coupling penalty $(1-\rho^2)^{-1} = 16\eta/(16\eta-9)$, which appears in both CRBs below.

\begin{theorem}[Exact CRB, $h$-known framework]
\label{thm:CRB}
With the channel gain $h$ known, the delay and Doppler CRBs are:
\begin{align}
\CRB(\tau_0) &= \frac{3\Delta t^2}{8\pi^2\SNR\, N}\cdot\frac{16}{16\eta-9}, \label{eq:crb_tau} \\
\CRB(f_D) &= \frac{3}{8\pi^2\SNR\, T^2 N}\cdot\frac{16\eta}{16\eta-9}. \label{eq:crb_f}
\end{align}
The corresponding range and velocity CRBs follow from $R = c\tau_0/2$ and $v = cf_D/(2f_c)$:
\begin{align}
\CRB(R) &= \frac{3c^2\Delta t^2}{32\pi^2\SNR\, N}\cdot\frac{16}{16\eta-9}, \label{eq:crb_R} \\
\CRB(v) &= \frac{3c^2}{32\pi^2 f_c^2\SNR\, T^2 N}\cdot\frac{16\eta}{16\eta-9}. \label{eq:crb_v}
\end{align}
\end{theorem}

\begin{proof}
Substituting $I_{\tau\tau} = \eta\cdot 8\pi^2\SNR\, N/(3\Delta t^2)$ and $(1{-}\rho^2)^{-1} = 16\eta/(16\eta{-}9)$ into \eqref{eq:crb_tau_general}:
\begin{equation}
\CRB(\tau_0) = \frac{1}{\eta\cdot\frac{8\pi^2\SNR\, N}{3\Delta t^2}}\cdot\frac{16\eta}{16\eta-9} = \frac{3\Delta t^2}{8\pi^2\SNR\, N}\cdot\frac{16}{16\eta-9}.
\end{equation}
The delay CRB benefits from the $1/\eta$ factor in $I_{\tau\tau}^{-1}$ \emph{and} the reduction in coupling inflation, producing $16/(16\eta{-}9)$ rather than $16\eta/(16\eta{-}9)$.

For the Doppler CRB, $I_{ff}$ does not contain $\eta$, so:
\begin{equation}
\CRB(f_D) = \frac{3}{8\pi^2\SNR\, T^2 N}\cdot\frac{16\eta}{16\eta-9}.
\end{equation}
\end{proof}

\begin{remark}[Asymmetric role of $\eta$]
\label{rmk:asymmetric_eta}
The delay CRB modification factor $16/(16\eta{-}9)$ is monotonically \emph{decreasing} in $\eta$ (larger $\eta$ improves delay CRB), while the Doppler factor $16\eta/(16\eta{-}9)$ is monotonically \emph{decreasing} toward~$1$ (larger $\eta$ has negligible effect on Doppler CRB). In the limit $\eta\gg 1$: $\CRB(\tau_0)\approx 3\Delta t^2/(8\pi^2\SNR\, N\eta)$ improves as $1/\eta$, while $\CRB(f_D)\approx 3/(8\pi^2\SNR\, T^2 N)$ becomes independent of $\eta$.
\end{remark}

\begin{remark}[Implication for ISAC waveform design]
\label{rmk:design_implication}
Theorem~\ref{thm:CRB} exposes a clean \emph{dimensional separation} between the two sensing parameters: the delay CRB depends on the chirp rate $c_1$ (through $\eta$), while the Doppler CRB does not. In ISAC practice, this means $c_1$ can be tuned purely for delay resolution---by maximizing the number of chirp wrapping cycles $C = 2Nc_1$ subject to the communication-driven diversity constraint~\cite{bemani2023twc}---without incurring a Doppler-estimation penalty. This is structurally distinct from single-chirp FMCW sensing, where bandwidth simultaneously controls range resolution and the dwell time constrains Doppler resolution. The separation originates from the cancellation mechanism of Theorem~\ref{thm:cancellation}: because $c_1$-dependent terms cancel in the delay score's cross-coupling with $f_D$, only $I_{\tau\tau}$ inherits the $\eta$-dependence.
\end{remark}

\subsection{$\phi$-Profiled Framework}
\label{subsec:phi_profiled}

When the channel phase $\phi$ is unknown, the parameter vector extends to $\bm{\theta} = (\alpha, \phi, \tau_0, f_D)$, and the $(\tau_0, f_D)$ effective FIM is obtained via the Schur complement
\begin{equation}
I_{\theta_i\theta_j}^{\mathrm{eff}} = I_{\theta_i\theta_j} - \frac{I_{\phi\theta_i}\cdot I_{\phi\theta_j}}{I_{\phi\phi}}, \qquad \theta_i,\theta_j\in\{\tau_0, f_D\}.
\label{eq:schur}
\end{equation}

Using the derivatives \eqref{eq:dy_dalpha}--\eqref{eq:dy_dphi} and the effective delay coefficient $A_{\mathrm{eff}}(m')$ from \eqref{eq:A_eff}, the cross-FIM elements are (see Appendix~\ref{app:fim} for details):
\begin{align}
I_{\phi\phi} &= 2\SNR\, N, \label{eq:I_phiphi} \\
I_{\phi\tau} &\approx \frac{2\pi\SNR\, N}{\Delta t}, \label{eq:I_phitau} \\
I_{\phi f} &= 2\pi\SNR\, TN. \label{eq:I_phif}
\end{align}
A key observation is that $I_{\phi\tau}$ does not depend on $c_1$, since the $c_1$-dependent terms in the delay derivative are eliminated by the cancellation (the effective coefficient $A_{\mathrm{eff}}$ is $c_1$-free).

\begin{theorem}[Exact decoupling after $\phi$-profiling]
\label{thm:decoupling}
After Schur complementation with respect to $\phi$, the effective cross-FIM element is exactly zero:
\begin{align}
I_{\tau f}^{\mathrm{eff}} &= I_{\tau f} - \frac{I_{\phi\tau}\cdot I_{\phi f}}{I_{\phi\phi}} \notag \\
&= 2\pi^2\SNR\, N^2 - \frac{\frac{2\pi\SNR\, N}{\Delta t}\cdot 2\pi\SNR\, TN}{2\SNR\, N} = 0.
\label{eq:I_tf_eff}
\end{align}
This holds for \emph{all} values of $c_1$, including $c_1 = 0$ (OFDM).
\end{theorem}

\begin{proof}
The second term evaluates to $2\pi^2\SNR\, TN/\Delta t = 2\pi^2\SNR\, N^2$ (using $T = N\Delta t$), which exactly equals $I_{\tau f}$. The cancellation holds because $I_{\tau f}$ is proportional to $(\sum m')(\sum n) \propto [\E(U)]^2$, while $I_{\phi\tau}I_{\phi f}/I_{\phi\phi}$ produces the same product through the Schur complement structure.
\end{proof}

\begin{remark}[Physical interpretation]
Profiling out $\phi$ is equivalent to a ``de-meaning'' operation on the $\Xi$ linear phase: the delay information shifts from the mean of $m'$ (which couples with Doppler) to the variance of $m'$ (which does not). This completely eliminates the delay-Doppler coupling, regardless of $c_1$.
\end{remark}

The effective diagonal FIM elements are:
\begin{align}
I_{\tau\tau}^{\mathrm{eff}} &= I_{\tau\tau} - \frac{I_{\phi\tau}^2}{I_{\phi\phi}} = \frac{2\pi^2\SNR\, N(4\eta{-}3)}{3\Delta t^2}, \label{eq:I_tt_eff} \\
I_{ff}^{\mathrm{eff}} &= I_{ff} - \frac{I_{\phi f}^2}{I_{\phi\phi}} = \frac{2\pi^2\SNR\, T^2 N}{3}. \label{eq:I_ff_eff}
\end{align}

Since $I_{\tau f}^{\mathrm{eff}} = 0$, the profiled CRBs have no coupling inflation:
\begin{align}
\CRB_{\mathrm{prof}}(\tau_0) &= \frac{1}{I_{\tau\tau}^{\mathrm{eff}}} = \frac{3\Delta t^2}{2\pi^2\SNR\, N(4\eta{-}3)}, \label{eq:crb_tau_prof} \\
\CRB_{\mathrm{prof}}(f_D) &= \frac{1}{I_{ff}^{\mathrm{eff}}} = \frac{3}{2\pi^2\SNR\, T^2 N}. \label{eq:crb_f_prof}
\end{align}

\begin{remark}[$\CRB_{\mathrm{prof}}(f_D)$ is independent of $c_1$]
The profiled Doppler CRB does not contain $\eta$ and is identical for all values of $c_1$, including OFDM ($c_1=0$). This is because $\eta$ affects only $I_{\tau\tau}$, and after profiling, the Doppler CRB is decoupled from the delay information.
\end{remark}

\subsection{Consistency with OFDM Literature}
\label{subsec:ofdm_consistency}

Setting $c_1 = 0$ gives $\eta = 1$ and $4\eta{-}3 = 1$, so
\begin{equation}
\CRB_{\mathrm{prof}}(\tau_0)\big|_{c_1=0} = \frac{3\Delta t^2}{2\pi^2\SNR\, N}.
\label{eq:crb_tau_ofdm}
\end{equation}
This result exactly matches the CRB derived by Gaudio~\emph{et al.}\ \cite{gaudio2019ofdm} for OFDM radar (Lemma~1, single frame, $M=N$ subcarriers), after converting from their normalized delay parameter $t = \Delta f\cdot\tau$ to physical delay $\tau_0 = t\cdot N\Delta t$:
\begin{equation}
\begin{split}
\CRB_{\text{Gaudio}}(\tau_0)
  &= (N\Delta t)^2 \cdot \frac{6}{4\pi^2 \SNR\, N^3} \\
  &= \frac{3\Delta t^2}{2\pi^2 \SNR\, N}
   = \CRB_{\mathrm{prof}}(\tau_0)\big|_{c_1=0}. \quad
\end{split}
\end{equation}
Similarly, the profiled Doppler CRB \eqref{eq:crb_f_prof} matches the OFDM result.

\subsection{Two-Framework Comparison}

Table~\ref{tab:two_frameworks} compares the two CRB frameworks.

\begin{table}[t]
\centering
\caption{Comparison of $h$-known and $\phi$-profiled CRB frameworks.}
\label{tab:two_frameworks}
\begin{tabular}{@{}lcc@{}}
\toprule
& \textbf{$h$-known} & \textbf{$\phi$-profiled} \\
\midrule
FIM dimension & $2\times 2$ & $4\times 4\to 2\times 2$ \\
$I_{\tau f}$ & $2\pi^2\SNR\, N^2$ & $0$ (exact) \\
$\rho$ & $3/(4\sqrt{\eta})$ & $0$ \\
\midrule
$\CRB(\tau_0)$ & $\frac{3\Delta t^2}{8\pi^2\SNR\, N}\cdot\frac{16}{16\eta-9}$ & $\frac{3\Delta t^2}{2\pi^2\SNR\, N(4\eta-3)}$ \\[0.8em]
$\CRB(f_D)$ & $\frac{3}{8\pi^2\SNR\, T^2 N}\cdot\frac{16\eta}{16\eta-9}$ & $\frac{3}{2\pi^2\SNR\, T^2 N}$ \\[0.5em]
\midrule
$\eta\gg 1$ limit & $\frac{3\Delta t^2}{8\pi^2\SNR\, N\eta}$ & $\frac{3\Delta t^2}{8\pi^2\SNR\, N\eta}$ (agrees) \\[0.5em]
\bottomrule
\end{tabular}
\end{table}

The two frameworks agree in the large-$\eta$ limit and share the same qualitative conclusion: larger $\eta$ (larger $c_1$) improves the delay CRB. The $h$-known framework has the advantage of explicitly revealing the coupling structure $\rho = 3/(4\sqrt{\eta})$, while the $\phi$-profiled framework provides complete decoupling and direct comparability with the OFDM literature.

\subsection{Chirp Parameter Dependence}
\label{subsec:c1_dep}

Since $\eta\propto C^2 = (2Nc_1)^2$ due to cyclic wrapping \eqref{eq:eta_total}, the profiled delay CRB scales as
\begin{equation}
\CRB_{\mathrm{prof}}(\tau_0) \approx \frac{\Delta t^2}{8\pi^2\SNR\,\lceil l{+}\iota\rceil\cdot C(C{-}2)} \propto c_1^{-2}
\label{eq:crb_tau_vs_c1}
\end{equation}
for $C\gg 1$, where $4\eta{-}3\approx 4\eta_{\mathrm{wrap}} \approx 12\lceil l{+}\iota\rceil C(C{-}2)/N \approx 12\lceil l{+}\iota\rceil C^2/N$ is used.

\begin{corollary}[Chirp parameter dependence]
\label{thm:c1_dep}
In the asymptotic regime $C\gg 1$:
\begin{enumerate}
\item Increasing $c_1$ \emph{significantly improves} the delay CRB (approximately as $c_1^{-2}$).
\item The Doppler CRB is \emph{entirely independent} of $c_1$ in both frameworks.
\end{enumerate}
Combining (1) with Corollary~\ref{cor:afdm_weaker}, we conclude that $\CRB_{\mathrm{AFDM}}(\tau_0) < \CRB_{\mathrm{OFDM}}(\tau_0)$ for any $c_1 \in (0, 1/2)$, with the gap widening as $c_1$ increases. The improvement originates from the cyclic wrapping effect: larger $C$ produces larger phase jumps at the circular boundary that carry delay-only information.
\end{corollary}

\begin{remark}[Finite-$C$ behavior]
The $c_1^{-2}$ scaling is the asymptotic result for $C\gg 1$. In the practical operating range $C\in[9,20]$, numerical fitting yields $\eta\approx 0.97C - 2.5$ (approximately linear), corresponding to $\CRB\propto 1/C$ rather than $1/C^2$. The two scaling laws produce comparable predictions in this range. Even at the communication-optimal $C = 2\alpha_{\max}{+}1 = 9$, the delay CRB is approximately 4 times better than OFDM.
\end{remark}

\begin{remark}[$N$-dependence under two operating modes]
\label{rmk:N_scaling}
The CRB scaling with $N$ depends critically on how $c_1$ is chosen. If $c_1$ is held fixed as $N$ increases, then $C = 2Nc_1 \propto N$ and $\eta_{\mathrm{wrap}} \propto C^2/N \propto N$, so the profiled delay CRB scales as $1/(N\cdot\eta) \propto N^{-2}$. If instead $C$ is held fixed (i.e., $c_1 = C/(2N) \propto 1/N$), then $\eta$ remains approximately constant and the CRB scales as $1/N$---the same rate as OFDM. Thus, the $\eta$-induced improvement is an $N$-dependent bonus that is realized only when the chirp parameter is not scaled down with increasing $N$.
\end{remark}

\subsection{Equivalent Delay CRB}

The equivalent delay $\lequiv$ can be estimated as a one-dimensional parameter via $\partial\mathcal{F}/\partial\lequiv = -(j2\pi/N)\mathcal{D}_{\mathrm{s}}$, yielding the Fisher information $I_{\lequiv} = 8\pi^2\SNR\, N/3$ and
\begin{equation}
\CRB(\lequiv) = \frac{3}{8\pi^2\SNR\, N}.
\label{eq:crb_leq}
\end{equation}
This one-dimensional CRB involves no coupling penalty and relates to the physical delay CRB via the mapping $\tau_0 = (\lequiv - Tf_D)\Delta t/(2Nc_1)$, providing a factor $C^2$ improvement when $f_D$ is known. The extraction of this $C^2$-fold potential information through structured search is addressed in a companion paper \cite{paper3_leq_extraction}.

\subsection{Using the Bounds in Practice}
\label{subsec:using_bounds}

The closed-form expressions derived in this section serve three concrete purposes for an ISAC system designer. First, they provide a \emph{benchmarking baseline}: any proposed AFDM sensing algorithm can be compared against \eqref{eq:crb_tau} and \eqref{eq:crb_f}, with the gap quantifying suboptimality without resorting to Monte-Carlo CRB simulation at every operating point. Second, they enable \emph{waveform parameter selection}. Because $\eta$ depends on the number of chirp wrapping cycles $C = 2Nc_1$ and is computable in closed form via \eqref{eq:eta_total}, the designer can select $c_1$ to meet a target delay accuracy without tuning by simulation; the asymmetric dependence (delay improves quadratically, Doppler is invariant) means this selection does not need to balance against Doppler accuracy. Third, they provide a \emph{regime-of-validity test} for the OFDM-based intuition that pervades the sensing literature: the proposed bounds reduce continuously to the classical OFDM CRB \cite{gaudio2019ofdm} as $c_1 \to 0$, so a designer can quantify exactly how much chirp rate is needed before AFDM-specific phenomena dominate over OFDM-like behavior. Worked numerical examples illustrating each of these uses are deferred to Section~\ref{sec:simulation}.

\section{Numerical Validation: From the Mechanism to the Bounds}
\label{sec:simulation}

This section provides comprehensive numerical verification of the results of Sections~\ref{sec:derivatives}--\ref{sec:crb}. The presentation deliberately mirrors the paper's causal chain: we first validate the cancellation mechanism itself (Section~\ref{subsec:sim_cancel}), then the closed-form FIM elements it produces (Section~\ref{subsec:sim_fim}), then the $3/(4\sqrt{\eta})$ coupling coefficient (Section~\ref{subsec:sim_rho}), and finally the $c_1^{-2}$ CRB gain (Section~\ref{subsec:sim_crb}). Throughout, numerical FIM and CRB values are computed exactly via central-difference differentiation of the $4\times 4$ FIM (including $\alpha$, $\phi$, $\tau_0$, $f_D$), with no analytical approximations; theoretical closed forms are plotted only for comparison.

\subsection{Simulation Setup}

Unless otherwise stated, the system parameters are: $N = 128$, $\Delta t = 52.08$~ns (corresponding to a 19.2~MHz bandwidth), $l{+}\iota = 1.31$ ($l=1$, $\iota=0.31$), $k{+}\kappa = 4.2$ ($k=4$, $\kappa=0.2$), and $\alpha_{\max}=4$. The chirp parameter $c_1$ is swept by varying $C = 2Nc_1$ from $1$ to $20$ over all integer values and 150 uniformly spaced non-integer values. Derivatives are computed by central differences with step size chosen such that the perturbation in $\lequiv$-space is approximately $10^{-4}$.

\subsection{Headline Experiment: Numerically Verifying the Cancellation Mechanism}
\label{subsec:sim_cancel}

We begin with the central experiment of this section: a direct numerical verification that Paths~2 and~3a in the delay score cancel as predicted by Theorem~\ref{thm:cancellation}. This is the structural observation on which all subsequent closed-form results depend, so we examine it before any FIM or CRB validation. Table~\ref{tab:cancel_verify} reports the relevant quantities at the communication-optimal configuration $C=9$ ($c_1 = 9/256$).

\begin{table}[t]
\centering
\caption{Numerical verification of the cancellation in Theorem~\ref{thm:cancellation}.}
\label{tab:cancel_verify}
\small
\begin{tabular}{@{}lcc@{}}
\toprule
\textbf{Quantity} & \textbf{Value} & \textbf{Meaning} \\
\midrule
$\|\partial\mathbf{H}/\partial\tau_0\|_F$ & $1.70\times 10^{9}$ & delay deriv. \\
$\|(2Nc_1/\Delta t)\partial\mathbf{H}/\partial\lequiv\|_F$ & $7.05\times 10^{9}$ & Path~2 pred. \\
$\|\text{residual}\|_F$ & $6.66\times 10^{9}$ & Path~3a + others \\
\midrule
cosine (Path~2, residual) & $-0.971$ & anti-parallel \\
delay / Path~2 ratio & $0.24$ & 76\% cancelled \\
\midrule
$|\mathrm{corr}(\partial\mathbf{H}/\partial\tau_0,\,\partial\mathbf{H}/\partial f_D)|$ & $0.344$ & weak coupling \\
$|\mathrm{corr}(\partial\mathbf{H}/\partial\lequiv,\,\partial\mathbf{H}/\partial f_D)|$ & $1.000$ & same kernel \\
\bottomrule
\end{tabular}
\end{table}

The Path~2 prediction and the residual (containing Path~3a) have comparable magnitudes ($7.05$ vs.\ $6.66 \times 10^9$) but nearly opposite directions (cosine angle $-0.971$), so they cancel to leave a true derivative with only $24\%$ of the Path~2 magnitude---a $76\%$ energy reduction. After cancellation, the delay-Doppler derivative correlation drops from $1.0$ (last row, pre-cancellation $\lequiv$ kernel) to $0.344$ (post-cancellation), confirming that the CPP compensation replaces the shared kernel $\mathcal{D}_{\mathrm{s}}$ with the structurally different $\mathcal{F}$. Without the CPP-induced Path~3a, the AFDM delay score would be dominated by the chirp-rate-scaled Doppler kernel and the FIM would be near-singular ($\rho\to 1$); the numerical cancellation observed here is therefore direct experimental evidence that the CPP simultaneously serves as a delay-Doppler decoupler, not just a channel equalizer.

\subsection{FIM Closed-Form Verification}
\label{subsec:sim_fim}

Having confirmed the cancellation mechanism itself, we now verify that the closed-form FIM elements predicted by Theorem~\ref{thm:FIM_closed} match the numerical FIM. We emphasize that prior AFDM-CRB studies either evaluated the FIM purely numerically~\cite{ranasinghe2025joint} or gave closed-form expressions only for the OFDM special case ($c_1 = 0$)~\cite{zhang2025afdm}; neither retained the staircase-residual contribution $\mathcal{D}_{\epsilon}$ that arises from the imperfect linear-chirp approximation of $Q(n,m)$. Theorem~\ref{thm:FIM_closed} keeps this term explicitly, which is what makes the closed form valid across the full range $c_1\in(0,1/2)$ rather than only at the OFDM endpoint. Table~\ref{tab:fim_verify} compares the closed-form FIM elements against numerical computation at $C=9$, $l{+}\iota=1.31$.

\begin{table}[t]
\centering
\caption{Closed-form versus numerical FIM elements and derived quantities.}
\label{tab:fim_verify}
\begin{tabular}{@{}lccc@{}}
\toprule
\textbf{Element} & \textbf{Closed form} & \textbf{Numerical} & \textbf{Error} \\
\midrule
$I_{\tau\tau}^{(\Xi)}/\SNR$ & $1.24\times 10^{18}$ & $1.24\times 10^{18}$ & ${<}0.1\%$ \\
$I_{ff}/\SNR$ & $1.50\times 10^{-7}$ & $1.48\times 10^{-7}$ & $1.2\%$ \\
$I_{\tau f}/\SNR$ & $3.23\times 10^{5}$ & $3.18\times 10^{5}$ & $1.7\%$ \\
$\eta$ (numerical) & --- & $4.66$ & --- \\
$\rho$ & $0.348$ & $0.344$ & $1.1\%$ \\
\bottomrule
\end{tabular}
\end{table}

The $\Xi$ component of $I_{\tau\tau}$ matches to within $0.1\%$, and both $I_{ff}$ and $I_{\tau f}$ agree to within $2\%$. The cross-term error does not propagate significantly to $\rho$: since $\rho = I_{\tau f}/\sqrt{I_{\tau\tau}I_{ff}}$ involves a square-root denominator, a $1.7\%$ error in $I_{\tau f}$ reduces to a $1.1\%$ error in $\rho$. The coupling coefficient $\rho = 0.344$ matches the closed form $3/(4\sqrt{\eta}) = 0.348$ to within $1.1\%$, and this level of agreement is representative of all integer $C$ values (see Fig.~\ref{fig:rho_vs_C}). The numerical $\eta = 4.66$ is substantially larger than the sawtooth-only estimate $\eta_{\mathrm{saw}} \approx 1.90$, confirming the dominance of cyclic wrapping even at the smallest communication-compatible $C$.

\subsection{Residual Inflation Factor $\eta$ vs.\ $C$}
\label{subsec:sim_eta}

The closed-form FIM concentrates all waveform-dependent structure in $\eta$. We now examine how $\eta$ scales with the number of chirp wrapping cycles $C$ and validate the approximate formula~\eqref{eq:eta_total}. Fig.~\ref{fig:eta_vs_C} compares three approaches for computing $\eta$: (i) numerical FIM extraction, (ii) the wrapping-corrected formula $\eta \approx \eta_{\mathrm{saw}} + \eta_{\mathrm{wrap}}$ from \eqref{eq:eta_total}, and (iii) the sawtooth-only approximation $\eta_{\mathrm{saw}} = 1 + 3\sigma_W^2$.

\begin{figure}[t]
\centering
\includegraphics[width=\columnwidth]{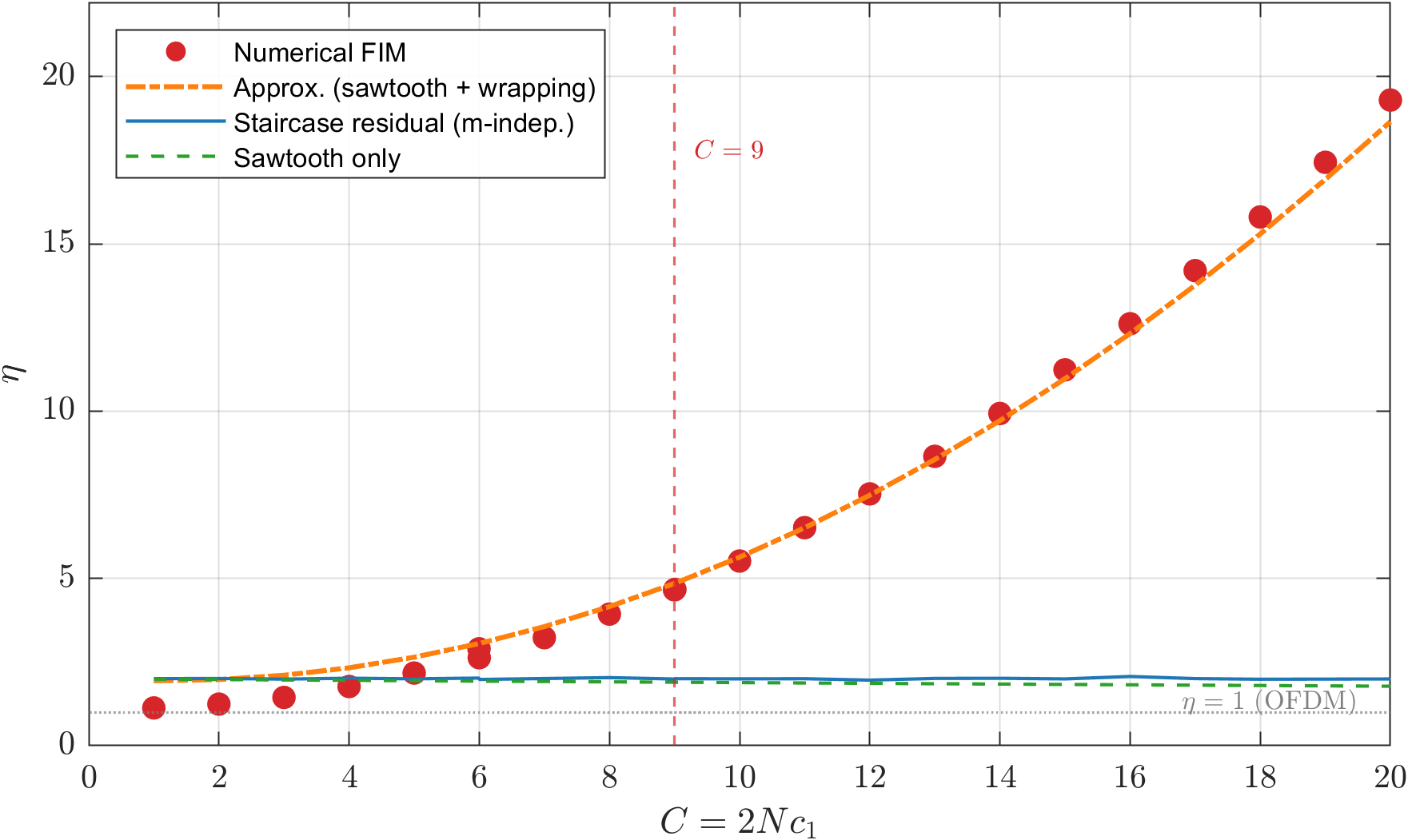}
\caption{Residual inflation factor $\eta$ versus $C$. Red circles: numerical FIM. Orange dashed: wrapping-corrected formula~\eqref{eq:eta_total}. Green dotted: sawtooth-only $1+3\sigma_W^2$. The numerical $\eta$ grows approximately as $C^2$, consistent with $\eta_{\mathrm{wrap}} \approx 3\lceil l{+}\iota\rceil C(C{-}2)/N$.}
\label{fig:eta_vs_C}
\end{figure}

The numerical $\eta$ increases from $1.1$ at $C=1$ to $19.3$ at $C=20$, far exceeding the sawtooth-only upper bound of $2$. The wrapping-corrected $C(C{-}2)$ formula captures the correct $C^2$ growth trend, with less than $4\%$ error in the communication-compatible region $C \geq 9$; for small $C \leq 5$, residual cross-term contributions become non-negligible and the formula overestimates $\eta$ by up to $22\%$ (see Appendix~\ref{app:eta} for the analysis). The sensing gain factor $4\eta{-}3$ ranges from $15.6$ ($C = 9$) to $74.2$ ($C = 20$), providing one to two orders of magnitude improvement over OFDM.

\subsection{Coupling Coefficient $\rho$ and Profiled Decoupling}
\label{subsec:sim_rho}

With the $\eta$ behavior established, we turn to the downstream consequence---the $3/(4\sqrt{\eta})$ coupling coefficient (Theorem~\ref{thm:rho}) and the exact profiled decoupling $I_{\tau f}^{\mathrm{eff}} = 0$ (Theorem~\ref{thm:decoupling}). Fig.~\ref{fig:rho_vs_C} verifies both.

\begin{figure}[t]
\centering
\includegraphics[width=\columnwidth]{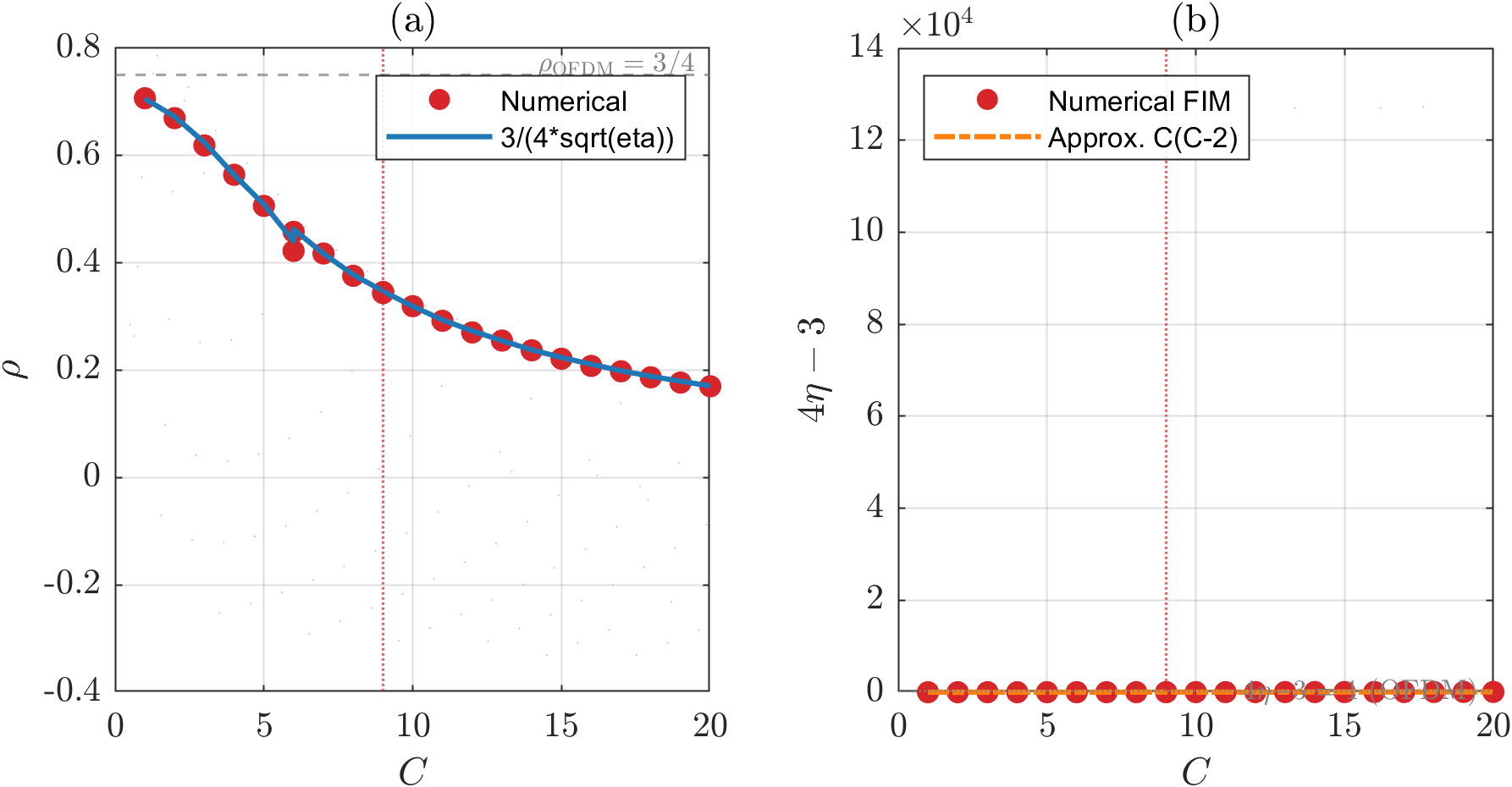}
\caption{Left: coupling coefficient $\rho$ versus $C$. Numerical values (circles) match the closed form $3/(4\sqrt{\eta})$ (dashed). Right: profiled cross-FIM ratio $|I_{\tau f}^{\mathrm{eff}}/I_{\tau\tau}^{\mathrm{eff}}|$, confirming the exact decoupling of Theorem~\ref{thm:decoupling}.}
\label{fig:rho_vs_C}
\end{figure}

The left panel confirms that AFDM's coupling is strictly weaker than OFDM for all $c_1 > 0$: at the communication-optimal $C = 9$, $\rho = 0.34$ corresponds to a coupling inflation $(1{-}\rho^2)^{-1} = 1.13$, i.e., only $13\%$ penalty for joint estimation. The right panel provides direct numerical evidence that Schur complementation eliminates the delay-Doppler coupling \emph{exactly}---the residual is at the level of floating-point arithmetic, not an approximation artifact.

\subsection{Delay CRB and the $c_1^{-2}$ Dimensional Gain}
\label{subsec:sim_crb}

The ultimate operational metric is the CRB itself. Fig.~\ref{fig:crb_tau_vs_C} shows the profiled delay CRB versus $C$ at three SNR levels, validating the closed-form expression~\eqref{eq:crb_tau_prof} and quantifying the gain over OFDM.

\begin{figure}[t]
\centering
\includegraphics[width=\columnwidth]{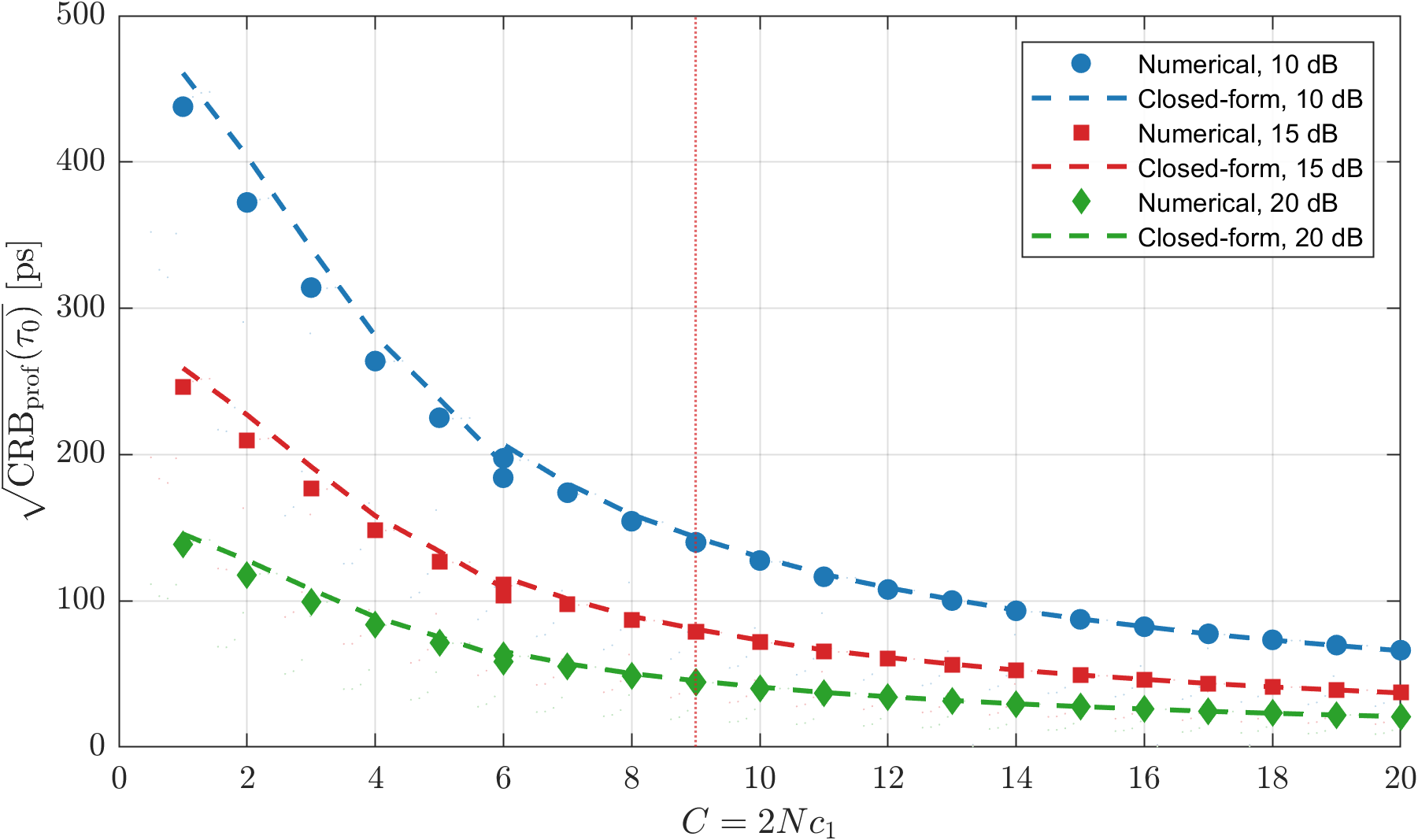}
\caption{$\sqrt{\CRB_{\mathrm{prof}}(\tau_0)}$ versus $C$ at $\SNR = 10$, $15$, and $20$~dB. Circles: numerical; dashed lines: closed form~\eqref{eq:crb_tau_prof} with numerically extracted $\eta$. CRB decreases monotonically with $C$.}
\label{fig:crb_tau_vs_C}
\end{figure}

At $\SNR = 15$~dB, the delay root-CRB decreases from $246$~ps at $C=1$ to $37$~ps at $C=20$, representing a factor of $6.6$ improvement. Even at the communication-optimal $C = 2\alpha_{\max}{+}1 = 9$, the AFDM delay CRB ($79$~ps) is approximately four times better than OFDM ($319$~ps). The closed-form CRB \eqref{eq:crb_tau_prof} with numerically extracted $\eta$ matches the numerical computation at all $C$ values.

Table~\ref{tab:crb_vs_C} lists the key sensing metrics at representative integer-$C$ operating points (with $\SNR=15$~dB and the communication-optimal $C=9$ highlighted in bold), providing a quantitative summary of the improvement over OFDM.

\begin{table}[t]
\centering
\caption{Sensing performance at selected integer-$C$ operating points.}
\label{tab:crb_vs_C}
\begin{tabular}{@{}ccccc@{}}
\toprule
$C$ & $\eta$ & $\rho$ & $\sqrt{\CRB(\tau_0)}$~[ps] & vs.\ OFDM \\
\midrule
1  & 1.13  & 0.71 & 246 & $0.77\times$ \\
5  & 2.17  & 0.51 & 127 & $0.40\times$ \\
\textbf{9}  & \textbf{4.66}  & \textbf{0.34} & \textbf{79}  & $\mathbf{0.25\times}$ \\
12 & 7.53  & 0.27 & 61  & $0.19\times$ \\
15 & 11.24 & 0.22 & 49  & $0.15\times$ \\
20 & 19.30 & 0.17 & 37  & $0.12\times$ \\
\midrule
OFDM & 1.00 & 0.75 & 319 & $1.00\times$ \\
\bottomrule
\end{tabular}
\end{table}

\subsection{Doppler CRB Independence of $c_1$}
\label{subsec:sim_doppler}

The delay CRB of the previous subsection improves as $c_1^{-2}$; Theorem~\ref{thm:CRB} predicts that the Doppler CRB should be completely insensitive to $c_1$. Fig.~\ref{fig:crb_f_flat} verifies this directly.

\begin{figure}[t]
\centering
\includegraphics[width=\columnwidth]{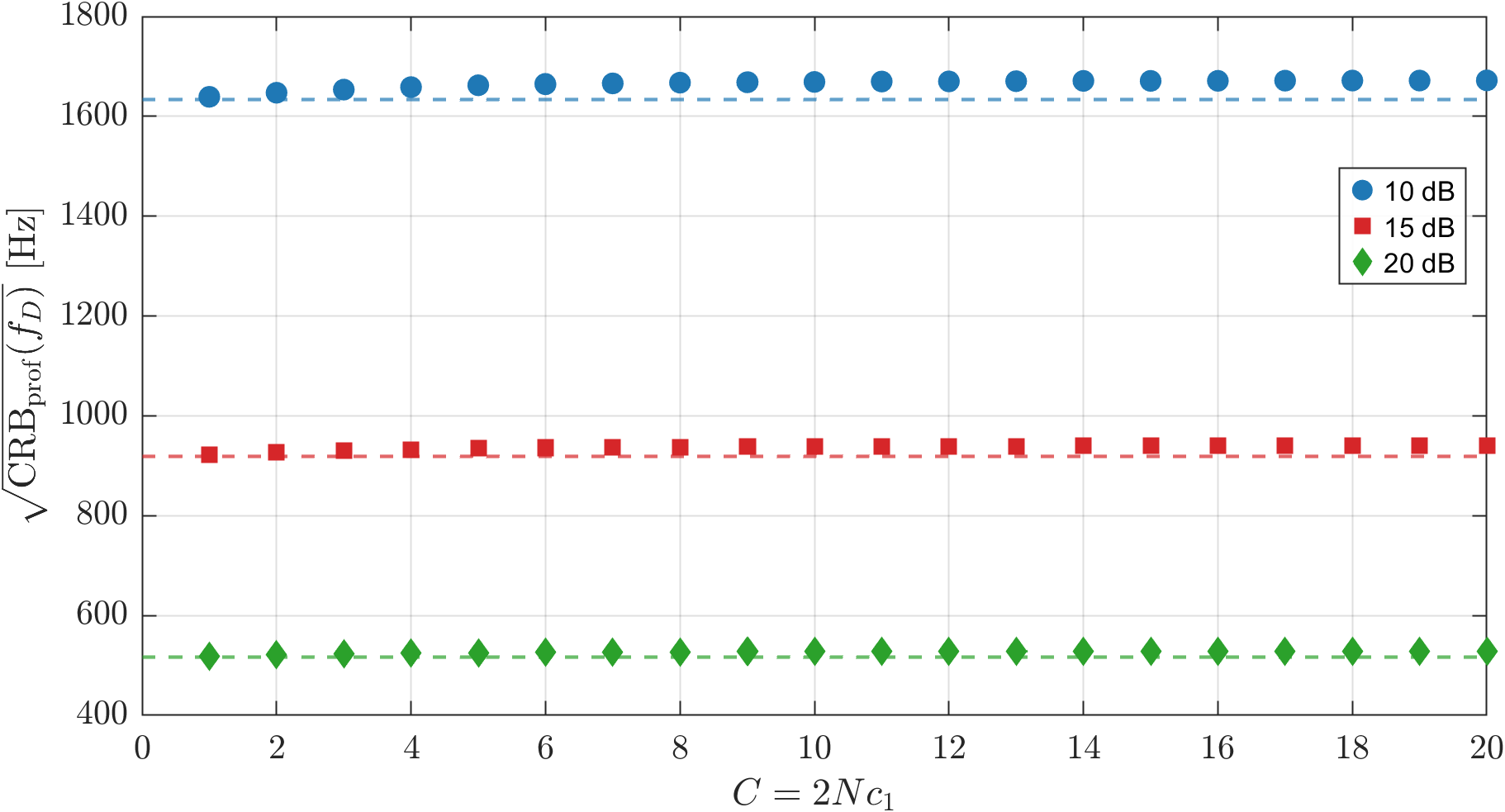}
\caption{$\sqrt{\CRB_{\mathrm{prof}}(f_D)}$ versus $C$ at $\SNR = 10$, $15$, and $20$~dB. All curves are flat (coefficient of variation ${<}0.5\%$), confirming that $I_{ff}^{\mathrm{eff}}$ does not contain $\eta$ and the profiled Doppler CRB~\eqref{eq:crb_f_prof} is identical for AFDM and OFDM.}
\label{fig:crb_f_flat}
\end{figure}

The numerical Doppler CRB falls precisely on the theoretical horizontal line given by \eqref{eq:crb_f_prof} at all three SNR levels. This confirms the structural asymmetry identified in Remark~\ref{rmk:asymmetric_eta}: increasing $c_1$ improves the delay CRB (via $\eta$) while leaving the Doppler CRB entirely unaffected.

\subsection{Robustness and Scaling Checks}
\label{subsec:sim_robust}

The preceding subsections have validated the main chain: cancellation mechanism, closed-form FIM, coupling coefficient, and CRBs. Two auxiliary checks complete the picture: the robustness of the closed-form results to the fractional delay $\iota$ (including the integer-delay limit $\iota=0$), and the scaling of the CRB with $N$ under two practical operating modes.

\subsubsection{Fractional delay sensitivity and integer-delay non-degeneracy}
\label{subsec:sim_iota}

Fig.~\ref{fig:iota_sweep} examines the sensitivity of $\eta$ and $\rho$ to the fractional delay $\iota$, with particular attention to the integer-delay case $\iota = 0$.

\begin{figure}[t]
\centering
\includegraphics[width=\columnwidth]{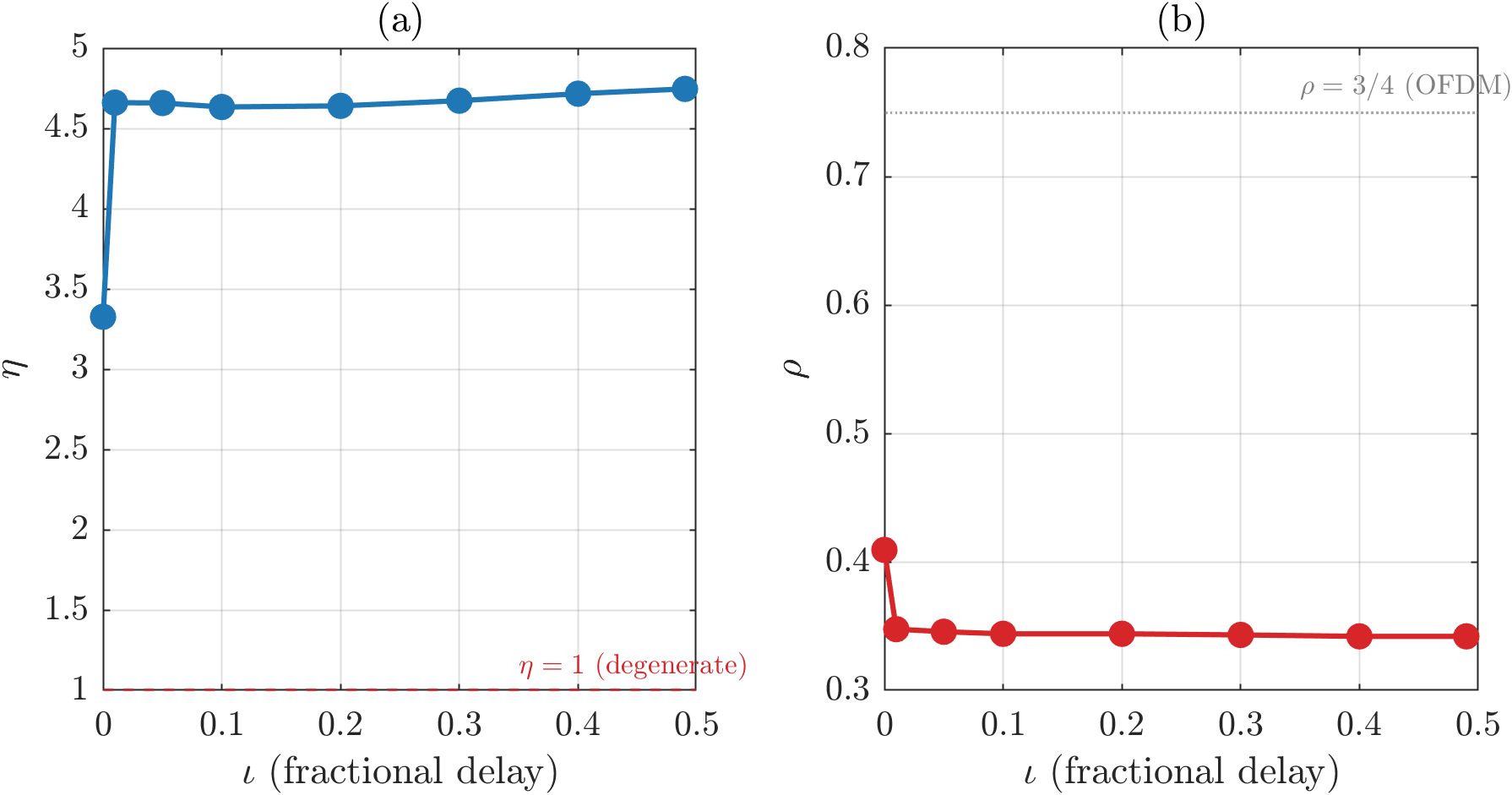}
\caption{$\eta$ and $\rho$ versus fractional delay $\iota$ ($C=9$). Both quantities are well-defined at $\iota=0$ (integer delay) and stabilize rapidly for $\iota > 0.01$.}
\label{fig:iota_sweep}
\end{figure}

At $\iota = 0$, the FIM has $\eta = 3.3 \gg 1$ and $\rho = 0.41 \ll 1$. This is because Path~3a ($\partial\Phi_{\mathrm{jump}}/\partial\tau_0|_{\iota=0} = (2\pi/\Delta t)Q(n,m)|_{\iota=0} \neq 0$) remains active even when $\Phi_{\mathrm{jump}}$ itself vanishes---the function value is zero but its derivative is not. Both $\eta$ and $\rho$ exhibit a sharp transition at $\iota \approx 0$ (due to the change in the number of wrapped samples $\lceil l{+}\iota\rceil$) but stabilize rapidly for $\iota > 0.01$. The cancellation mechanism is therefore robust to fractional-delay variations across the operationally relevant range.

\subsubsection{$N$ scaling laws under two operating modes}
\label{subsec:sim_N}

Fig.~\ref{fig:N_scaling} examines how the CRB scales with $N$ under the two operating modes identified in Remark~\ref{rmk:N_scaling}.

\begin{figure}[t]
\centering
\includegraphics[width=\columnwidth]{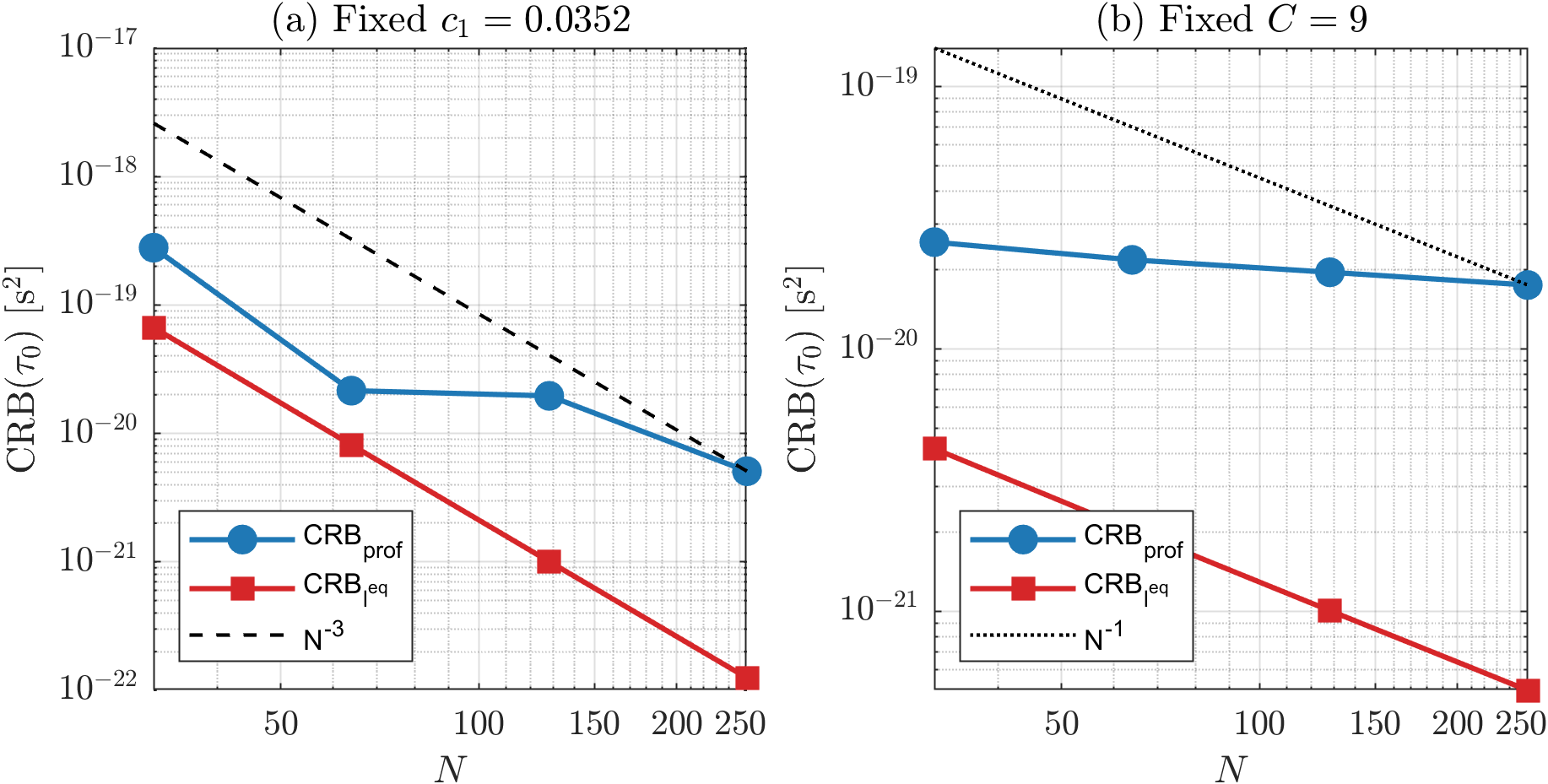}
\caption{$N$ scaling laws ($\SNR=10$~dB, $N \in \{32,64,128,256\}$). Left: fixed $c_1 = 9/256$ ($C\propto N$); CRB$_{\mathrm{prof}}(\tau_0)$ scales approximately as $N^{-1.8}$. Right: fixed $C=9$ ($c_1\propto 1/N$); CRB$_{\mathrm{prof}}(\tau_0)$ scales as $N^{-0.2}$.}
\label{fig:N_scaling}
\end{figure}

\subsubsection{Fixed $c_1$ ($C\propto N$)} When $c_1$ is held constant, $C$ grows linearly with $N$, so $\eta_{\mathrm{wrap}}\propto C^2/N \propto N$. The profiled CRB then scales as $1/(N\cdot\eta) \propto N^{-2}$. Numerical fitting yields $N^{-1.8}$, slightly slower than $N^{-2}$ because the constant $\eta_{\mathrm{saw}}$ term is not yet negligible at finite $N$.

\subsubsection{Fixed $C$ ($c_1\propto 1/N$)} When $C$ is held constant, $\eta_{\mathrm{wrap}}\propto C^2/N \propto 1/N$, so $\eta$ decreases with $N$. The CRB scales as $1/(N\cdot\eta) \approx 1/(aN + b)$, yielding approximately $N^{-0.2}$ in the numerical fit---a much slower improvement. In practical systems where $c_1 = (2\alpha_{\max}{+}1)/(2N)$ (fixed $C$), increasing $N$ primarily improves frequency resolution rather than the delay CRB.

\section{Conclusion}
\label{sec:conclusion}

This paper has addressed a structural question raised by recent work on AFDM-based ISAC: why, and by how much, is AFDM's delay-Doppler estimation more decoupled than OFDM's? The answer lies in a single structural observation---the CPP phase jump exactly cancels the chirp-induced frequency drift in the delay score function, up to an $O(1/C)$ residual captured by the scalar $\eta$.

Every closed-form result in this paper is a consequence of that cancellation. The $2\times 2$ Fisher information matrix parameterizes as a function of $\eta$ alone (Theorem~\ref{thm:FIM_closed}); the delay-Doppler coupling coefficient $3/(4\sqrt{\eta})$ is strictly below the $3/4$ of OFDM for any nonzero chirp rate (Theorem~\ref{thm:rho}); the delay CRB scales as $c_1^{-2}$ while the Doppler CRB is $c_1$-independent (Theorem~\ref{thm:CRB}); and Schur complementation yields an exact delay-Doppler decoupling that reduces continuously to the classical OFDM CRB of~\cite{gaudio2019ofdm} as $c_1 \to 0$ (Theorem~\ref{thm:decoupling}). Viewed through this lens, the CPP---previously studied only as a channel-equalization device---emerges as the structural element that shapes AFDM's sensing estimability.

The framework has three immediate implications for ISAC waveform design. First, because the two CRBs depend on the chirp rate asymmetrically (delay improving quadratically, Doppler unaffected), the chirp rate emerges as a pure sensing-resolution knob that does not trade off against Doppler estimation, simplifying joint sensing-communication parameter selection. Second, even at the communication-optimal configuration $C = 2\alpha_{\max}{+}1$, the AFDM delay CRB is approximately four times better than OFDM---the gain is available without sacrificing communication performance. Third, the continuous reduction to OFDM as $c_1 \to 0$ certifies that the proposed framework is a strict generalization of the classical OFDM radar sensing theory of~\cite{gaudio2019ofdm}, rather than a parallel track.

Several extensions merit discussion. The present analysis assumes a single point target; the multi-target case introduces inter-target interference that may degrade the FIM structure when targets share similar equivalent delays. The closed-form approximation for $\eta$ is accurate ($<4\%$ error) in the operationally relevant region $C \geq 9$ where AFDM is typically deployed; in the small-$C$ regime ($C \leq 5$), additional cross-terms become non-negligible and direct numerical FIM evaluation is recommended (see Appendix~\ref{app:eta}). The analysis also assumes integer $C = 2Nc_1$, which ensures perfect CPP alignment; non-integer $C$ introduces additional residual effects not captured by the current framework. Natural extensions include the design of structured estimators that exploit the equivalent-delay information identified in this paper, the joint optimization of the chirp rate for sensing-communication trade-offs, and the generalization to MIMO-AFDM and multi-frame processing scenarios.

\appendices

\section{Proof of Derivative Decomposition}
\label{app:derivatives}

This appendix provides the detailed derivations supporting Section~\ref{sec:derivatives}: the four-path decomposition, the linear approximation of the staircase function, the cancellation proof, and the effective coefficient derivation.

\subsection{Dependency Analysis of $\tau_0$}

From the observation model \eqref{eq:io_relation}, $\tau_0$ enters through $l{+}\iota = \tau_0/\Delta t$. The intermediate quantities and their dependencies are:
\begin{itemize}
\item $h = \alpha e^{j\phi}$: independent of $\tau_0$, so $\partial h/\partial\tau_0 = 0$.
\item $\Xi(m,m') = c_1(l{+}\iota)^2 - c_2(m^2{-}m'^2) - (l{+}\iota)m'/N$: depends on $\tau_0$ through $(l{+}\iota)$.
\item $\lequiv = (k{+}\kappa) + 2Nc_1(l{+}\iota)$: depends on $\tau_0$ through $2Nc_1(l{+}\iota)/\Delta t$.
\item $\Phi_{\mathrm{jump}}(n,m) = 2\pi\iota\cdot Q(n,m;l{+}\iota)$: depends on $\tau_0$ through \emph{both} $\iota$ (scaling factor) \emph{and} $Q$ (through segment boundaries $n_{m,q}$).
\end{itemize}

\subsection{Path~1: Detailed Derivation}

Differentiating $\Xi$ in \eqref{eq:Xi_def} term by term:
\begin{align}
\frac{\partial\Xi}{\partial\tau_0} &= \frac{\partial}{\partial\tau_0}\!\left[c_1(l{+}\iota)^2 - c_2(m^2{-}m'^2) - \frac{(l{+}\iota)m'}{N}\right] \notag \\
&= c_1\cdot 2(l{+}\iota)\cdot\frac{1}{\Delta t} - 0 - \frac{m'}{N\Delta t} \notag \\
&= \frac{1}{\Delta t}\!\left[2c_1(l{+}\iota) - \frac{m'}{N}\right]\!.
\label{eq:app_dXi_dtau}
\end{align}
The contribution to $\partial y[m]/\partial\tau_0$ via the product rule on $e^{j2\pi\Xi}\cdot\mathcal{F}$ is:
\begin{equation}
\frac{\partial e^{j2\pi\Xi}}{\partial\tau_0}\cdot\mathcal{F} = j2\pi\frac{\partial\Xi}{\partial\tau_0}\cdot e^{j2\pi\Xi}\cdot\mathcal{F} = A(m')\cdot e^{j2\pi\Xi}\cdot\mathcal{F}
\end{equation}
with the coefficient $A(m')$ defined in \eqref{eq:coeff_A}.

\subsection{Paths~2 and~3: Product Rule on $\Phi_{\mathrm{jump}}$}

The total phase of $\Gamma_n$ is $\varphi_n = 2\pi(m'{-}m{-}\lequiv)n/N + 2\pi\iota\cdot Q(n,m)$. The derivative $\partial\varphi_n/\partial\tau_0$ receives contributions from both terms.

\paragraph{Path~2} The Dirichlet kernel phase depends on $\tau_0$ only through $\lequiv$:
\begin{align}
\frac{\partial}{\partial\tau_0}\!\left[\frac{2\pi(m'{-}m{-}\lequiv)n}{N}\right] &= -\frac{2\pi n}{N}\cdot\frac{\partial\lequiv}{\partial\tau_0} \notag \\
&= -\frac{2\pi n}{N}\cdot\frac{2Nc_1}{\Delta t} = -\frac{4\pi c_1 n}{\Delta t}.
\end{align}

\paragraph{Path~3a and~3b} Since $\Phi_{\mathrm{jump}} = 2\pi\iota\cdot Q$ is a product of two $\tau_0$-dependent factors, the product rule gives:
\begin{equation}
\frac{\partial\Phi_{\mathrm{jump}}}{\partial\tau_0} = 2\pi\underbrace{\frac{\partial\iota}{\partial\tau_0}}_{=\,1/\Delta t}\cdot Q + 2\pi\iota\cdot\frac{\partial Q}{\partial\tau_0}.
\label{eq:app_product_rule}
\end{equation}
The first term (Path~3a) yields $(2\pi/\Delta t)Q(n,m)$, acting on all $N$ samples with magnitude $O(c_1 N)$. The second term (Path~3b) involves the derivative of the indicator functions at the $C$ segment boundaries:
\begin{equation}
\frac{\partial Q}{\partial\tau_0} = \frac{1}{\Delta t}\sum_{q=0}^{C}q\cdot\frac{\partial\mathcal{I}_{L_{m,q}}}{\partial(l{+}\iota)},
\end{equation}
which is nonzero only at the $C$ boundary points and has magnitude $O(\iota)$, much smaller than Path~3a.

Combining all three contributions into \eqref{eq:dGamma_dtau} and then applying the product rule to $e^{j2\pi\Xi}\cdot\mathcal{F}$ yields the complete four-term expression \eqref{eq:dy_dtau_complete}.

\subsection{Staircase Function Analysis}
\label{app:staircase}

\begin{proof}[Proof of Proposition~\ref{prop:Q_linear}]
The staircase function $Q(n,m) = \sum_{q=0}^{C}q\cdot\mathcal{I}_{L_{m,q}}(\tilde{n})$ assigns the value $q$ to the $q$-th segment of width $W = N/C$. Viewing $Q$ as a staircase approximation to the continuous linear function $f(\tilde{n}) = 2c_1\tilde{n} = (C/N)\tilde{n}$:
\begin{itemize}
\item At the left endpoint of segment $q$ ($\tilde{n} = n_{m,q}$): $f(n_{m,q}) \approx q$, so $\epsilon_Q \approx 0$.
\item At the right endpoint ($\tilde{n} = n_{m,q+1}{-}1$): $f(n_{m,q+1}{-}1) \approx q{+}1{-}2c_1$, so $\epsilon_Q \approx -(1{-}2c_1) \approx -1$.
\end{itemize}
Within each segment, $Q = q$ is constant while $2c_1\tilde{n}$ increases linearly, so $\epsilon_Q = Q - 2c_1\tilde{n}$ decreases linearly from $\approx 0$ to $\approx -1$---a sawtooth waveform.

The amplitude satisfies $|\epsilon_Q| \leq 1$. The root-mean-square value within one segment is
\begin{equation}
\sqrt{\frac{1}{W}\sum_{k=0}^{W-1}\!\left(\frac{k}{W}\right)^{\!2}} = \sqrt{\frac{(W{-}1)(2W{-}1)}{6W^2}} \to \frac{1}{\sqrt{3}} \approx 0.29
\end{equation}
for large $W$. Since $Q_{\max} = C = 2c_1 N$, the relative error is $\|\epsilon_Q\|_{\mathrm{rms}}/Q_{\max} = O(1/C) \ll 1$.
\end{proof}

\subsection{Cancellation Proof}

\begin{proof}[Complete proof of Theorem~\ref{thm:cancellation}]
The sum of Paths~2 and~3a within $\partial\varphi_n/\partial\tau_0$ is
\begin{equation}
S(n) \triangleq -2c_1 n + Q(n,m) = -2c_1 n + 2c_1\tilde{n} + \epsilon_Q(n,m).
\label{eq:app_S_def}
\end{equation}
Consider two cases based on whether $n$ is in the non-wrapped or wrapped region.

\textit{Case 1: $n \geq \lceil l{+}\iota\rceil$ (non-wrapped).} Here $\tilde{n} = n - (l{+}\iota)$, so
\begin{equation}
-2c_1 n + 2c_1\tilde{n} = -2c_1 n + 2c_1(n - (l{+}\iota)) = -2c_1(l{+}\iota).
\end{equation}
This is a constant independent of $n$.

\textit{Case 2: $n < \lceil l{+}\iota\rceil$ (wrapped).} Here $\tilde{n} = n - (l{+}\iota) + N$, so
\begin{equation}
-2c_1 n + 2c_1\tilde{n} = -2c_1(l{+}\iota) + 2c_1 N = -2c_1(l{+}\iota) + C.
\end{equation}
This differs from Case~1 by an additive constant $C$.

In both cases, the $n$-dependent terms ($-2c_1 n$ and $2c_1 n$ from $2c_1\tilde{n}$) cancel exactly, leaving an $n$-independent offset plus the sawtooth residual $\epsilon_Q$. Before cancellation, each of $-2c_1 n$ and $Q(n,m)$ ranges over $[0, C]$; after cancellation, $S(n)$ consists of a constant $-2c_1(l{+}\iota)$ (or $-2c_1(l{+}\iota)+C$ in the wrapped region) plus $\epsilon_Q \in [-1,1]$. The magnitude reduction ratio is $O(1/C)$.
\end{proof}

\subsection{Effective Coefficient Derivation}

After cancellation, the combined Paths~2+3a contribution merges with Path~1. Define the post-cancellation kernel decomposition:
\begin{align}
B_{\mathrm{s}}\mathcal{D}_{\mathrm{s}} + B_{\iota}\mathcal{D}_{\iota} &= \frac{j2\pi}{\Delta t}\sum_n S(n)\Gamma_n \notag \\
&\approx \frac{j2\pi}{\Delta t}\!\left[-2c_1(l{+}\iota)\mathcal{D}_{\Xi} + \mathcal{D}_{\epsilon}\right]
\end{align}
where the wrapping correction is absorbed into $\mathcal{D}_{\epsilon}$ for notational simplicity.

Adding the constant term $-j4\pi c_1(l{+}\iota)\mathcal{D}_{\Xi}/\Delta t$ to Path~1's contribution $A(m')\mathcal{D}_{\Xi}$:
\begin{align}
A_{\mathrm{eff}}(m') &= A(m') - \frac{j4\pi c_1(l{+}\iota)}{\Delta t} \notag \\
&= j2\pi\!\left(\frac{2c_1(l{+}\iota)}{\Delta t} - \frac{m'}{N\Delta t}\right) - \frac{j4\pi c_1(l{+}\iota)}{\Delta t} \notag \\
&= \frac{j4\pi c_1(l{+}\iota)}{\Delta t} - \frac{j2\pi m'}{N\Delta t} - \frac{j4\pi c_1(l{+}\iota)}{\Delta t} \notag \\
&= -\frac{j2\pi m'}{N\Delta t}.
\label{eq:app_Aeff}
\end{align}
The $c_1(l{+}\iota)$ terms cancel \emph{exactly} in the third line, yielding the $c_1$-independent effective coefficient stated in \eqref{eq:A_eff}. This exact cancellation holds regardless of the values of $c_1$, $l$, and $\iota$.

\section{FIM Ten-Term Decomposition}
\label{app:fim}

This appendix provides the detailed derivations for Section~\ref{sec:fim}: the Parseval lemma proof, the kernel inner products, the ten-term magnitude analysis, and the FIM-level cancellation proof.

\subsection{Proof of the Parseval Lemma}

\begin{proof}[Proof of Lemma~\ref{lemma:parseval}]
Under the approximation that $\Phi_{\mathrm{jump}}(n,m)$ is independent of $m$ (denoted $\Phi_{\mathrm{jump}}(n)$), define
\begin{equation}
\begin{aligned}
f_n &= a(n)\,e^{j2\pi(m'-\lequiv)n/N}\,e^{j\Phi_{\mathrm{jump}}(n)}, \\
g_n &= b(n)\,e^{j2\pi(m'-\lequiv)n/N}\,e^{j\Phi_{\mathrm{jump}}(n)}.
\end{aligned}
\end{equation}
Then $\mathcal{A}(m,m') = \sum_n f_n e^{-j2\pi mn/N}$ and $\mathcal{B}(m,m') = \sum_n g_n e^{-j2\pi mn/N}$ are the $N$-point DFTs of $\{f_n\}$ and $\{g_n\}$ with respect to $m$. By Parseval's theorem:
\begin{equation}
\sum_{m=0}^{N-1}\mathcal{A}^*(m,m')\mathcal{B}(m,m') = N\sum_{n=0}^{N-1}f_n^* g_n.
\end{equation}
Expanding $f_n^* g_n$: since the phase factors $e^{\pm j2\pi(m'-\lequiv)n/N}$ and $e^{\pm j\Phi_{\mathrm{jump}}(n)}$ have unit modulus and appear conjugated,
\begin{equation}
f_n^* g_n = a^*(n)\,b(n).
\end{equation}
All phase factors cancel in the conjugate product, leaving a result independent of $m'$, $\lequiv$, and $\Phi_{\mathrm{jump}}$.
\end{proof}

\subsection{Key Kernel Inner Products}

Applying Lemma~\ref{lemma:parseval} with appropriate weights:

\begin{enumerate}[label=\textit{\alph*}), leftmargin=1.5em, itemsep=3pt]
\item $\mathcal{F}$ \textit{self-inner-product} ($a = b = 1$):\;
  $\sum_m|\mathcal{F}|^2 = N\sum_n 1 = N^2$.

\item $\mathcal{D}_{\mathrm{s}}$ \textit{self-inner-product} ($a = b = n$):\;
  $\sum_m|\mathcal{D}_{\mathrm{s}}|^2 = N\sum_n n^2
  = N\cdot N(N{-}1)(2N{-}1)/6 \approx N^4/3$.

\item $\mathcal{F}$--$\mathcal{D}_{\mathrm{s}}$ \textit{cross}
  ($a = 1,\, b = n$):\;
  $\sum_m\mathcal{F}^*\mathcal{D}_{\mathrm{s}}
  = N\sum_n n = N\cdot N(N{-}1)/2 \approx N^3/2$.

\item $\mathcal{D}_{\iota}$ \textit{self-inner-product}
  ($a = b = Q \approx 2c_1\tilde{n}$):\;
  $\sum_m|\mathcal{D}_{\iota}|^2 \approx N\sum_n(2c_1\tilde{n})^2
  = 4c_1^2 N\sum_n\tilde{n}^2 \approx 4c_1^2 N^4/3$,
  where $\sum_n\tilde{n}^2 = \sum_n n^2$ by the circular shift property.

\item $\mathcal{D}_{\mathrm{s}}$--$\mathcal{D}_{\iota}$ \textit{cross}
  ($a = n,\, b = Q \approx 2c_1(n{-}(l{+}\iota))$):
  \begin{align}
  \sum_m\mathcal{D}_{\mathrm{s}}^*\mathcal{D}_{\iota}
    &\approx 2c_1 N\!\left[\sum_n n^2 - (l{+}\iota)\sum_n n\right]
    \notag \\
    &\approx 2c_1 N\!\left[\frac{N^3}{3}
      - (l{+}\iota)\frac{N^2}{2}\right]\!.
  \end{align}
  The leading term is $2c_1 N^4/3$.

\item $\mathcal{D}_{\epsilon}$ \textit{self-inner-product}
  ($a = b = \epsilon_Q(n,m)$):\;
  Since $\epsilon_Q$ depends on~$m$, the Parseval lemma applies only to
  the $m'$-sum for fixed~$m$:
  $\sum_{m'}|\mathcal{D}_{\epsilon}(m,m')|^2
  = N\sum_n\epsilon_Q^2(n,m) = N\cdot E_{\epsilon}(m)$.
  Summing over~$m$:
  $\sum_{m,m'}|\mathcal{D}_{\epsilon}|^2 = N\sum_m E_{\epsilon}(m)$.
\end{enumerate}.

\subsection{Ten-Term Definitions and Magnitude Analysis}

The four self-terms and six cross-terms are defined in \eqref{eq:Sigma_res}ff.\ of the main text. Their magnitudes, obtained by combining the coefficient modulus-squares with the kernel inner products above, are:

\begin{align}
\Sigma_{\Xi\Xi} &= \sum_{m'}\!|A(m')|^2\!\sum_m\!|\mathcal{F}|^2 \notag \\
&\approx \frac{4\pi^2}{\Delta t^2}\cdot\frac{N}{3}\cdot N^2 = \frac{4\pi^2 N^3}{3\Delta t^2}, \\
\Sigma_{\mathrm{s}} &= |B_{\mathrm{s}}|^2\cdot N\cdot\frac{N^4}{3} = \frac{16\pi^2 c_1^2 N^5}{3\Delta t^2}, \\
\Sigma_{\iota} &\approx |B_{\iota}|^2\cdot N\cdot 4c_1^2\frac{N^4}{3} = \frac{16\pi^2 c_1^2 N^5}{3\Delta t^2}, \\
\Sigma_{\mathrm{b}} &\approx |B_{\mathrm{b}}|^2\cdot N\cdot C^2 N^2 = \frac{16\pi^2\iota^2 c_1^2 N^5}{\Delta t^2}.
\end{align}
For $\Sigma_{\Xi\Xi}$, the sum $\sum_{m'}|A(m')|^2$ is evaluated using $\sum_{m'}(m'/N - 2c_1(l{+}\iota))^2 \approx N/3$ (the variance of $m'/N$ over the uniform distribution dominates). The key observation is $\Sigma_{\mathrm{s}} = \Sigma_{\iota}$ to leading order.

The dominant cross-term is:
\begin{align}
\Sigma_{\mathrm{s}\text{-}\iota} &= 2\Re\!\left\{B_{\mathrm{s}} B_{\iota}^*\sum_{m,m'}\!\mathcal{D}_{\mathrm{s}}\mathcal{D}_{\iota}^*\right\}\!.
\end{align}
Computing $B_{\mathrm{s}} B_{\iota}^* = (-j4\pi c_1/\Delta t)(-j2\pi/\Delta t) = -8\pi^2 c_1/\Delta t^2$ (negative real) and using $\sum_{m,m'}\mathcal{D}_{\mathrm{s}}\mathcal{D}_{\iota}^* \approx N\cdot 2c_1 N^4/3$:
\begin{align}
\Sigma_{\mathrm{s}\text{-}\iota} &= 2\!\left(-\frac{8\pi^2 c_1}{\Delta t^2}\right)\!\!\left(\frac{2c_1 N^5}{3}\right) \notag \\
&= -\frac{32\pi^2 c_1^2 N^5}{3\Delta t^2}.
\end{align}

\subsection{FIM-Level Cancellation: Detailed Proof}

\begin{proof}[Detailed proof of Theorem~\ref{thm:FIM_cancellation}]
Summing the three dominant terms:
\begin{align}
&\Sigma_{\mathrm{s}} + \Sigma_{\iota} + \Sigma_{\mathrm{s}\text{-}\iota} \notag \\
&\quad = \frac{16\pi^2 c_1^2 N^5}{3\Delta t^2} + \frac{16\pi^2 c_1^2 N^5}{3\Delta t^2} - \frac{32\pi^2 c_1^2 N^5}{3\Delta t^2} = 0.
\end{align}
Alternatively, recognizing $\Sigma_{\mathrm{s}} + \Sigma_{\iota} + \Sigma_{\mathrm{s}\text{-}\iota} = \sum_{m,m'}|u{+}v|^2$ with $u = B_{\mathrm{s}}\mathcal{D}_{\mathrm{s}}$, $v = B_{\iota}\mathcal{D}_{\iota}$, the cancellation follows directly from Theorem~\ref{thm:cancellation}: the $n$-linear parts of $u{+}v$ vanish, collapsing $|u{+}v|^2$ from $O(c_1^2 N^5/\Delta t^2)$ to $O(N^3/\Delta t^2)$.
\end{proof}

\subsection{Cross-Term Cancellation: $\Sigma_{\Xi\text{-}\mathrm{s}} + \Sigma_{\Xi\text{-}\iota}$}

The two intermediate cross-terms combine as
\begin{equation}
\Sigma_{\Xi\text{-}\mathrm{s}} + \Sigma_{\Xi\text{-}\iota} = 2\Re\!\left\{\sum_{m,m'}\!A(m')\mathcal{D}_{\Xi}(u^*{+}v^*)\right\}\!.
\end{equation}
Since $u{+}v$ after cancellation is approximately $({j2\pi}/{\Delta t})[-2c_1(l{+}\iota)\mathcal{D}_{\Xi} + \mathcal{D}_{\epsilon}]$, and $A(m') = j\cdot(\text{real})$ while the coefficient $(-j2\pi/\Delta t)^* = j2\pi/\Delta t$ is also purely imaginary, the overall factor is $j\cdot j = -1$ times real quantities. However, $\sum_m\mathcal{D}_{\Xi}^*\mathcal{D}_{\Xi} = N^2$ (real) and $\sum_m\mathcal{D}_{\Xi}^*\mathcal{D}_{\epsilon}$ involves $\sum_n\epsilon_Q(n,m)$ which is real. Therefore, the product $A(m')\cdot(j2\pi/\Delta t)\cdot(\text{real})$ is $j\cdot j\cdot(\text{real}) = -(\text{real})$, and $\Re\{-(\text{real})\}$ is nonzero in general but of magnitude $O(N^3/\Delta t^2)$---the same order as $\Sigma_{\Xi\Xi}$.

In the idealized case where $\Phi_{\mathrm{jump}}$ is $m$-independent, additional phase symmetries force $\Sigma_{\Xi\text{-}\mathrm{s}} + \Sigma_{\Xi\text{-}\iota} \approx 0$. In practice, this sum contributes a small correction that is captured by the numerical computation of $\eta$.

\subsection{Derivation of $I_{\tau f}$}

The cross-FIM element involves the inner product of the delay and Doppler derivatives. Using the effective delay derivative \eqref{eq:dy_dtau_effective} (dominant $\Xi$ term) and the Doppler derivative \eqref{eq:dy_df}:
\begin{align}
I_{\tau f}^{(\Xi)} &= \frac{2\SNR}{N^2}\!\sum_{m,m'}\!\Re\!\left\{\!\left(\frac{-j2\pi m'}{N\Delta t}\right)^{\!\!*}\!\!\left(\frac{-j2\pi T}{N}\right)\mathcal{F}^*\mathcal{D}_{\mathrm{s}}\right\}\!.
\end{align}
The coefficient product is $(j2\pi m'/(N\Delta t))\cdot(-j2\pi T/N) = 4\pi^2 m'T/(N^2\Delta t)$, which is a positive real number. Using $\sum_m\Re\{\mathcal{F}^*\mathcal{D}_{\mathrm{s}}\} = N^3/2$ and $\sum_{m'}m' = N^2/2$:
\begin{align}
I_{\tau f}^{(\Xi)} &= \frac{2\SNR}{N^2}\cdot\frac{4\pi^2 T}{N^2\Delta t}\cdot\frac{N^3}{2}\cdot\frac{N^2}{2} \notag \\
&= \frac{2\pi^2\SNR\, TN}{\Delta t} = 2\pi^2\SNR\, N^2,
\end{align}
where $T/\Delta t = N$ is used in the final step.

The residual contribution $I_{\tau f}^{(\epsilon)}$ involves $B_{\iota}^*\cdot(-j2\pi T/N)$ multiplied by $\sum_{m,m'}\mathcal{D}_{\epsilon}^*\mathcal{D}_{\mathrm{s}}$. The coefficient product $(-j2\pi/\Delta t)(-j2\pi T/N) = -4\pi^2 T/(N\Delta t)$ is negative real. The kernel inner product $\sum_m\mathcal{D}_{\epsilon}^*\mathcal{D}_{\mathrm{s}} = N\sum_n n\epsilon_Q(n,m) = NC_{\epsilon}(m)$ where $C_{\epsilon}(m)$ is the correlation between $\epsilon_Q$ and the time index $n$. After summation over $m'$ (factor $N$) and $m$, this contributes $O(\SNR\, N^2)$, which is smaller than $I_{\tau f}^{(\Xi)} = O(\SNR\, N^2)$ by a factor depending on $C_{\epsilon}$. For the closed-form analysis, this correction is neglected.

\subsection{Derivation of $I_{\phi\tau}$ and Schur Complement}

The cross-FIM between $\phi$ and $\tau_0$ uses $\partial y/\partial\phi = jh\cdot(\text{signal})/N$ and the effective delay derivative. The key product is $j\cdot A_{\mathrm{eff}}(m') = j\cdot(-j2\pi m'/(N\Delta t)) = 2\pi m'/(N\Delta t)$, which is a positive real number. Hence:
\begin{align}
I_{\phi\tau} &= \frac{2\SNR}{N^2}\cdot\frac{2\pi}{N\Delta t}\sum_{m'}m'\!\sum_m|\mathcal{F}|^2 \notag \\
&= \frac{2\SNR}{N^2}\cdot\frac{2\pi}{N\Delta t}\cdot\frac{N^2}{2}\cdot N^2 = \frac{2\pi\SNR\, N}{\Delta t}.
\end{align}
Since $A_{\mathrm{eff}}$ does not depend on $c_1$, neither does $I_{\phi\tau}$.

For the Schur complement:
\begin{align}
\frac{I_{\phi\tau}\cdot I_{\phi f}}{I_{\phi\phi}} &= \frac{(2\pi\SNR\, N/\Delta t)(2\pi\SNR\, TN)}{2\SNR\, N} \notag \\
&= \frac{2\pi^2\SNR\, TN}{\Delta t} = 2\pi^2\SNR\, N^2 = I_{\tau f}.
\end{align}
Thus $I_{\tau f}^{\mathrm{eff}} = 0$.

\section{Residual Inflation Factor $\eta$: Closed Form and Cyclic Wrapping}
\label{app:eta}

This appendix derives the closed-form approximation for $\eta$ and analyzes the cyclic wrapping contribution.

\subsection{Exact Definition}

From Definition~\ref{def:eta}:
\begin{equation}
\eta = 1 + \frac{3}{N^2}\sum_{m=0}^{N-1}E_{\epsilon}(m), \quad E_{\epsilon}(m) = \sum_{n=0}^{N-1}\epsilon_Q^2(n,m),
\end{equation}
where $\epsilon_Q(n,m) = Q(n,m) - 2c_1\tilde{n}$ and $\tilde{n} = (n{-}(l{+}\iota))_N$.

\subsection{Segment-Internal Contribution}

Within segment $q$ (samples satisfying $\tilde{n}\in[n_{m,q}, n_{m,q+1})$ and $n \geq \lceil l{+}\iota\rceil$), the residual is $\epsilon_Q = q - 2c_1\tilde{n}$, decreasing linearly from $\approx 0$ to $\approx -1$. The per-segment sum of squared residuals is:
\begin{equation}
\sum_{k=0}^{W-1}\!\left(\frac{k}{W}\right)^{\!2} = \frac{(W{-}1)(2W{-}1)}{6W} \approx \frac{W}{3}.
\end{equation}
Across all $C$ segments in the non-wrapped region ($N - \lceil l{+}\iota\rceil$ samples), the contribution to $E_{\epsilon}(m)$ is approximately $C\cdot W\cdot\sigma_W^2 \approx N\sigma_W^2$, where
\begin{equation}
\sigma_W^2 = \frac{(W{-}1)(2W{-}1)}{6W^2}.
\end{equation}
Since this is approximately the same for all $m$:
\begin{equation}
\eta_{\mathrm{saw}} = 1 + \frac{3}{N^2}\cdot N\cdot N\sigma_W^2 = 1 + 3\sigma_W^2.
\end{equation}

\paragraph{Limiting behavior:}
\begin{itemize}
\item $W\to\infty$ (small $c_1$): $\sigma_W^2\to 1/3$, so $\eta_{\mathrm{saw}}\to 2$.
\item $W = 1$ ($c_1 = 1/2$): $\sigma_W^2 = 0$, so $\eta_{\mathrm{saw}} = 1$. Each segment has only one sample, making the staircase exactly linear.
\item $W = 2$ ($c_1 = 1/4$): $\sigma_W^2 = (1)(3)/(6\cdot 4) = 1/8$, so $\eta_{\mathrm{saw}} = 1 + 3/8 = 1.375$.
\item $c_1 = 0$ (OFDM): $C = 0$, $Q = 0$, $\epsilon_Q = 0$, $\eta = 1$.
\end{itemize}

\subsection{Cyclic Wrapping Contribution}
\label{app:wrapping}

For the $n_w = \lceil l{+}\iota\rceil$ samples with $n < \lceil l{+}\iota\rceil$, the circular shift gives $\tilde{n} = n - (l{+}\iota) + N$, so
\begin{equation}
2c_1\tilde{n} = 2c_1(n - (l{+}\iota) + N) = 2c_1(n - (l{+}\iota)) + C.
\end{equation}

In the non-wrapped region, $S(n) = -2c_1 n + Q(n,m) \approx -2c_1(l{+}\iota) + \epsilon_Q(n,m)$ with $|\epsilon_Q|\leq 1$. In the wrapped region, from the proof of Theorem~\ref{thm:cancellation} (Case~2):
\begin{equation}
S(n) = -2c_1(l{+}\iota) + C + \epsilon_Q(n,m).
\end{equation}
The excess squared residual per wrapped sample relative to the non-wrapped baseline is
\begin{equation}
(C + \epsilon_Q)^2 - \epsilon_Q^2 = C^2 + 2C\epsilon_Q.
\label{eq:excess_expand}
\end{equation}
The cross-term $2C\epsilon_Q$ cannot be neglected. Since the wrapped samples lie near the end of a staircase segment (the segment boundary at $n = \lceil l{+}\iota\rceil$ is reached just before wrapping), the residual at these points satisfies $\epsilon_Q \approx -1$. Substituting into \eqref{eq:excess_expand} gives an excess of $C^2 - 2C = C(C{-}2)$ per wrapped sample. For the $n_w$ wrapped samples across $N$ values of $m$:
\begin{equation}
\Delta E \approx N\cdot n_w\cdot C(C{-}2).
\end{equation}
Normalizing via $\eta = 1 + 3\sum_m E_{\epsilon}(m)/N^2$:
\begin{equation}
\eta_{\mathrm{wrap}} \approx \frac{3n_w C(C{-}2)}{N} = \frac{3\lceil l{+}\iota\rceil C(C{-}2)}{N}.
\end{equation}

\subsection{Combined Formula}

\begin{equation}
\eta \approx \underbrace{1 + 3\sigma_W^2}_{\eta_{\mathrm{saw}}} + \underbrace{\frac{3\lceil l{+}\iota\rceil C(C{-}2)}{N}}_{\eta_{\mathrm{wrap}}}.
\end{equation}

\paragraph{Growth behavior:} Since $\eta_{\mathrm{saw}} < 2$ is bounded, while $\eta_{\mathrm{wrap}} \propto C^2$ for large $C$, the wrapping component dominates for $C \geq 5$ in typical configurations. The overall $\eta\propto C^2$ scaling is the basis for the $\CRB(\tau_0)\propto c_1^{-2}$ dependence derived in Section~\ref{sec:crb}.

\paragraph{Physical interpretation:} The wrapping effect arises because the CPP's staircase structure creates a phase discontinuity of magnitude $\propto C$ at the circular boundary $n = 0$ (relative to the shifted origin $\tilde{n} = 0$). This discontinuity is purely delay-dependent and provides $I_{\tau\tau}$ with information that grows as $C^2$, while leaving $I_{ff}$ and $I_{\tau f}$ unchanged. Larger $c_1$ (more chirp wrapping) amplifies this effect, simultaneously increasing $\eta$ and decreasing $\rho$.

\paragraph{Accuracy:} In the communication-compatible region $C\geq 9$, the $C(C{-}2)$ formula approximates $\eta$ to within $4\%$. For small $C$ ($\leq 5$), the formula overestimates $\eta$ by up to $22\%$ because the Parseval-based factorization of $I_{\tau\tau}$ assumes that $\Phi_{\mathrm{jump}}$ is approximately independent of $m$; when the number of staircase segments is small, this approximation breaks down and the cross-term $I_{\tau\tau}^{(\mathrm{cross})}$ becomes non-negligible. This cross-term is negative, partially cancelling the residual contribution and reducing the true $\eta$ below the approximate value. Despite this limitation for small $C$, the formula correctly captures the $\eta\propto C^2$ scaling law and the physical origin of the inflation factor. For precise values at any $C$, numerical FIM computation is recommended.


\end{document}